\documentclass{elsarticle}
\usepackage{ProbGen}
\usepackage{hyperref}
\usepackage{tikz}
\usetikzlibrary{arrows,automata,shapes}
\biboptions{sort&compress}

\hypersetup{citecolor=blue!70!red,urlcolor=black!70!blue,linkcolor=blue!70!red,colorlinks=true}
\numberwithin{equation}{section}

\begin{document}

\begin{frontmatter}

\title{Coalgebraic Tools for\\Randomness-Conserving Protocols}

\author{Dexter Kozen}
\ead{kozen@cs.cornell.edu}

\author{Matvey Soloviev}
\ead{msoloviev@cs.cornell.edu}

\address{Cornell University}

\begin{abstract}
We propose a coalgebraic model for constructing and reasoning about state-based protocols that implement efficient reductions among random processes. We provide basic tools that allow efficient protocols to be constructed in a compositional way and analyzed in terms of the tradeoff between state and loss of entropy. We show how to use these tools to construct various entropy-conserving reductions between processes.
\end{abstract}


\end{frontmatter}


\section{Introduction}
\label{sec:intro}

\noindent
In low-level performance-critical computations---for instance, data-forwarding devices in packet-switched networks---it is often desirable to minimize local state in order to achieve high throughput. But if the situation requires access to a source of randomness, say to implement randomized routing or load-balancing protocols, it may be necessary to convert the output of the source to a form usable by the protocol. As randomness is a scarce resource to be conserved like any other, these conversions should be performed as efficiently as possible and with a minimum of machinery.

In this paper we propose a coalgebraic model for constructing and reasoning about state-based protocols that implement efficient reductions among random processes. By ``efficient'' we mean with respect to loss of entropy. Entropy is a measure of the amount of randomness available in a random source. For example, a fair coin generates entropy at the rate of one random bit per flip; a fair six-sided die generates entropy at the rate of $\log 6\approx 2.585$ random bits per roll. We view randomness as a limited computational resource to be conserved, like time or space.

Unfortunately, converting from one random source to another generally involves a loss of entropy, as measured by the ratio of the rate of entropy produced to the rate of entropy consumed. This quantity is called the \emph{efficiency} of the conversion protocol. For example, if we wish to simulate a coin flip by rolling a die and declaring heads if the number on the die is even and tails if it is odd, then the ratio of entropy production to consumption is $1/2.585 \approx .387$, so we lose about .613 bits of entropy per trial. The efficiency cannot exceed the information-theoretic bound of unity, but we would like it to be as close to unity as can be achieved with simple state-based devices. For example, we could instead roll the die and if the result is 1, 2, 3, or 4, output two bits 00, 01, 10, or 11, respectively---the first bit can be used now and the second saved for later---and if the result is 5 or 6, output a single bit 0 or 1, respectively. The efficiency is much better, about .645.

In this paper we introduce a coalgebraic model for the analysis of reductions between discrete processes. A key feature of the model is that it facilitates compositional reasoning. In \S3 we prove several results that show how the efficiency and state complexity of a composite protocol depend on the same properties of its constituent parts. This allows efficient protocols to be constructed and analyzed in a compositional way. We are able to cover a full range of input and output processes while preserving asymptotic guarantees about the relationship between memory use and conservation of entropy.

In \S4 we use the model to construct the following reductions between processes, where $k$ is a tunable parameter roughly proportional to the logarithm of the size of the state space:
\begin{itemize}
\item
$d$-uniform to $c$-uniform with efficiency $1 - \Theta(k^{-1})$;
\item
$d$-uniform to arbitrary rational with efficiency $1 - \Theta(k^{-1})$;
\item
$d$-uniform to arbitrary with efficiency $1 - \Theta(k^{-1})$;
\item
arbitrary to $c$-uniform with efficiency $1 - \Theta(\log k/k)$;
\item
$(1/r,(r-1)/r)$ to $c$-uniform with efficiency $1 - \Theta(k^{-1})$.
\end{itemize}
Thus choosing a larger value of $k$ (that is, allowing more state) results in greater efficiency, converging to the optimal of 1 in the limit. Here ``$d$-uniform'' refers to an independent and identically distributed (\iid) process that produces a sequence of letters from an alphabet of size $d$, each chosen independently with uniform probability $1/d$; ``arbitrary rational'' refers to an \iid\ process with arbitrary rational probabilities; and ``arbitrary'' refers to an \iid\ process with arbitrary real probabilities. In the last item, the input distribution is a coin flip with bias $1/r$. The notation $\Theta(\cdot)$ is the usual notation for upper and lower asymptotic bounds. These results quantify the dependence of efficiency on state complexity and give explicit bounds on the asymptotic rates of convergence to the optimal.

\subsection{Related Work}

Since von Neumann's classic paper showing how to simulate a fair coin with a coin of unknown bias \cite{vonNeumann51}, many authors have studied variants of this problem. Our work is heavily inspired by the work of Elias \cite{Elias72}, who studies entropy-optimal generation of uniform distributions from known sources. The definition of conservation of entropy is given there.

Mossel, Peres, and Hillar \cite{MosselPeresHillar05} show that there is a finite-state protocol to simulate a $q$-biased coin with a $p$-biased coin when $p$ is unknown if and only if $q$ is a rational function of $p$.

Peres \cite{Peres92} shows how to iterate von Neumann's procedure for producing a fair coin from a biased coin to approximate the entropy bound. Blum \cite{Blum85} shows how to extract a fair coin from a Markov chain. 

Another line of work by Pae and Loui \cite{PaeLoui06,PaeLoui05,Pae05} focuses on emitting samples from a variety of rational distributions given input from an unknown distribution, as in von Neumann's original problem. In \cite{PaeLoui06}, the authors introduce a family of von-Neumann-like protocols that approach asymptotic optimality as they consume more input symbols before producing output, and moreover can be shown to be themselves optimal among all such protocols. 

In \cite{DBLP:journals/tit/HanH97}, Han and Hoshi present a family of protocols for converting between arbitrary known input and output distributions, based on an interval-refinement approach. These protocols exhibit favorable performance characteristics and are comparable to the ones we present according to multiple metrics, but require an infinite state space to implement.

Finally, there is a large body of related work on extracting randomness from weak random sources (e.g.\ \cite{Nisan96,Nisan99,Ta-shma96,SrinivasanZuckerman99,Dodis04}). These models typically work with imperfect knowledge of the input source and provide only approximate guarantees on the quality of the output. Here we assume that the statistical properties of the input and output are known completely, and simulations must be exact.

\section{Definitions}
\label{sec:def}

A \emph{(discrete) random process} is a finite or infinite sequence of discrete random variables. We will view the process as producing a stream of letters from some finite alphabet $\Sigma$. We will focus mostly on \emph{independent and identically distributed (\iid) processes}, in which successive letters are generated independently according to a common distribution on $\Sigma$.

Informally, a \emph{reduction} from a random process $X$ with alphabet $\Sigma$ to another random process $Y$ with alphabet $\Gamma$ is a deterministic protocol that consumes a stream of letters from $\Sigma$ and produces a stream of letters from $\Gamma$. To be a valid reduction, if the letters of the input stream are distributed as $X$, then the letters of the output stream must be distributed as $Y$. In particular, for \iid\ processes $X$ and $Y$ in which the letters are generated independently according to distributions $\mu$ on $\Sigma$ and $\nu$ on $\Gamma$, respectively, we say that the protocol is a reduction from $\mu$ to $\nu$. Most (but not all) of the protocols considered in this paper will be finite-state.

To say that the protocol is \emph{deterministic} means that the only source of randomness is the input process. It makes sense to talk about the expected number of input letters read before halting or the probability that the first letter emitted is $a$, but any such measurements are taken with respect to the distribution on the space of inputs.

There are several ways to formalize the notion of a reduction. One approach, following \cite{Peres92}, is to model a reduction as a map $f:\Ss\to\Gs$ that is monotone with respect to the prefix relation on strings; that is, if $x,y\in\Ss$ and $x$ is a prefix of $y$, then $f(x)$ is a prefix of $f(y)$. Monotonicity implies that $f$ can be extended uniquely by continuity to domain $\Ss\union\So$ and range $\Gs\union\Gamma^\omega$. The map $f$ would then constitute a reduction from the random process $X=X_0X_1X_2\ldots$ to $f(X_0X_1X_2\ldots)=Y_0Y_1Y_2\ldots$, where the random variable $X_i$ gives the $i$th letter of the input stream and $Y_i$ the $i$th letter of the output stream. To be a reduction from $\mu$ to $\nu$, it must hold that if the $X_i$ are independent and identically distributed as $\mu$, then the $Y_i$ are independent and identically distributed as $\nu$.

In this paper we propose an alternative state-based approach in which protocols are modeled as coalgebras $\dd:S\times\Sigma\to S\times\Gs$, where $S$ is a (possibly infinite) set of \emph{states}.\footnote{This is a coalgebra with respect to the endofunctor $(-\times\Gs)^\Sigma$ on $\Set$. Normally, as the structure map for such a coalgebra, $\delta$ would be typed as $\delta:S\to (S\times\Gs)^\Sigma$, but we have recurried it here to align more with the intuition of $\delta$ as the transition map of an automaton.} We can view a protocol as a deterministic stream automaton with output. In each step, depending on its current state, the protocol samples the input process, emits zero or more output letters, and changes state, as determined by its transition function $\delta$. The state-based approach has the advantage that it is familiar to computer scientists, is easily programmable, and supports common constructions such as composition.

\subsection{Protocols and Reductions}
\label{sec:redprot}

Let $\Sigma$, $\Gamma$ be finite alphabets. Let $\Ss$ denote the set of finite words and $\So$ the set of $\omega$-words (streams) over $\Sigma$. We use $x,y,\ldots$ for elements of $\Ss$ and $\alpha,\beta,\ldots$ for elements of $\So$. The symbols $\prefeq$ and $\pref$ denote the prefix and proper prefix relations, respectively.

If $\mu$ is a probability measure on $\Sigma$, we endow $\So$ with the product measure in which each symbol is distributed as $\mu$. The notation $\Pr(A)$ for the probability of an event $A$ refers to this measure. The measurable sets of $\So$ are the Borel sets of the Cantor space topology whose basic open sets are the \emph{intervals} $\set{\alpha\in\So}{x\pref\alpha}$ for $x\in\Ss$, and $\mu(\set{\alpha\in\So}{x\pref\alpha})=\mu(x)$, where $\mu(a_1a_2\cdots a_n)=\mu(a_1)\mu(a_2)\cdots\mu(a_n)$; see \cite{Halmos50}.

A \emph{protocol} is a coalgebra $(S,\dd)$ where $\dd:S\times\Sigma\to S\times\Gs$. Intuitively, $\delta(s,a) = (t,x)$ means that in state $s$, it consumes the letter $a$ from its input source, emits a finite, possibly empty string $x$, and transitions to state $t$.

We can immediately extend $\dd$ to domain $S\times\Ss$ by coinduction:
\begin{align*}
\dd(s,\eps) &= (s,\eps)\\
\dd(s,ax) &= \letin{(t,y)}{\dd(s,a)}{\letin{(u,z)}{\dd(t,x)}{(u,yz)}}.
\end{align*}
Since the two functions agree on $S\times\Sigma$, we use the same name. It follows that
\begin{align*}
\dd(s,xy) &= \letin{(t,z)}{\dd(s,x)}{\letin{(u,w)}{\dd(t,y)}{(u,zw)}}.
\end{align*}
By a slight abuse, we define the \emph{length} of the output as the length of its second component as a string in $\Gs$ and write $\len{\dd(s,x)}$ for $\len{z}$, where $\dd(s,x)=(t,z)$.

A protocol $\dd$ also induces a partial map $\ds:S\times\So\pfun\Go$ by coinduction:\footnote{The definition is coinductive in the sense that it involves the greatest fixpoint of a monotone map. We must take the greatest fixpoint to get the infinite behaviors as well as the finite behaviors.}
\begin{align*}
\ds(s,a\alpha) &= \letin{(t,z)}{\dd(s,a)}{z\cdot\ds(t,\alpha)}.
\end{align*}
It follows that
\begin{align*}
\ds(s,x\alpha) &= \letin{(t,z)}{\dd(s,x)}{z\cdot\ds(t,\alpha)}.
\end{align*}
Given $\alpha\in\So$, this defines a unique infinite string in $\ds(s,\alpha)\in\Gamma^\omega$ except in the degenerate case in which only finitely many output letters are ever produced.

A protocol is said to be \emph{productive} (with respect to a given probability measure on input streams) if, starting in any state, an output symbol is produced within finite expected time. It follows from this assumption that infinitely many output letters are produced with probability 1. The supremum over all states $s$ of the expected time before an output symbol is produced starting from $s$ is called the \emph{latency} of the protocol. We will restrict attention to protocols with finite latency.

Now let $\nu$ be a probability measure on $\Gamma$. Endow $\Gamma^\omega$ with the product measure in which each symbol is distributed as $\nu$, and define
\begin{align*}
\nu(a_1a_2\cdots a_n)=\nu(a_1)\nu(a_2)\cdots\nu(a_n),\ \ a_i\in\Gamma.
\end{align*}
We say that a protocol $(S,\dd,s)$ with start state $s\in S$ is a \emph{reduction from $\mu$ to $\nu$} if for all $y\in\Gs$,
\begin{align}
\Pr(y \prefeq \ds(s,\alpha)) &= \nu(y),\label{def:red}
\end{align}
where the probability $\Pr$ is taken with respect to the product measure $\mu$ on $\So$.
This implies that the symbols of $\ds(s,\alpha)$ are independent and identically distributed as $\nu$.

\subsection{Restart Protocols}
\label{sec:restart}

A \emph{prefix code} is a subset $A\subs\Ss$ such that every element of $\So$ has at most one prefix in $A$. Thus the elements of a prefix code are pairwise $\prefeq$-incomparable. A prefix code is \emph{exhaustive} (with respect to a given probability measure on input streams) if $\Pr(\text{$\alpha\in\So$ has a prefix in $A$})=1$. By K\"onig's lemma, if every $\alpha\in\So$ has a prefix in $A$, then $A$ is finite and exhaustive, but exhaustive codes need not be finite; for example, under the uniform measure on binary streams, the prefix code $\set{0^n1}{n\ge 0}$ is infinite and exhaustive. 

We often think of prefix codes as representing their infinite extensions. By a slight abuse of notation, if $\mu$ is a probability measure on $\So$ and $A\subs\Ss$ is a prefix code, we define
\begin{align}
\mu(A) &= \mu(\set{\alpha\in\So}{\exists x\in A\ x\prec\alpha}).\label{eq:prefixcode}
\end{align}

A \emph{restart protocol} is a protocol $(S,\dd,s)$ of a special form determined by a function $f:A\to\Gs$, where $A$ is an exhaustive prefix code, $A\ne\{\eps\}$, and $s$ is a designated start state. Intuitively, starting in $s$, we read symbols of $\Sigma$ from the input stream until encountering a string $x\in A$, output $f(x)$, then return to $s$ and repeat. Note that we are not assuming $A$ to be finite.

Formally, we can take the state space to be
\begin{align*}
S &= \set{u\in\Ss}{\text{$x\not\prefeq u$ for any $x\in A$}}
\end{align*}
and define $\dd:S\times\Sigma\to S\times\Gs$ by
\begin{align*}
\dd(u,a) &= \begin{cases}
(ua,\eps), & ua\not\in A,\\
(\eps,z), & ua\in A \text{ and } f(ua)=z
\end{cases}
\end{align*}
with start state $\eps$. Then for all $x\in A$, $\dd(\eps,x) = (\eps,f(x))$.

As with the more general protocols, we can extend to a partial function on streams, but here the definition takes a simpler form:
\begin{align*}
\ds(\eps,x\alpha) &= f(x)\cdot\ds(\eps,\alpha),\quad x\in A,\ \alpha\in\So.
\end{align*}

A restart protocol is \emph{positive recurrent} (with respect to a given probability measure on input streams) if, starting in the start state $s$, the expected time before the next visit to $s$ is finite. All finite-state restart protocols are positive recurrent, but infinite-state ones need not be.

If a restart protocol is positive recurrent, then the probability of eventually restarting is $1$, but the converse does not always hold. For example, consider a restart protocol that reads a sequence of coin flips until seeing the first heads. If the number of flips it read up to that point is $n$, let it read $2^n$ more flips and output the sequence of all flips it read, then restart. The probability of restarting is 1, but the expected time before restarting is infinite.

\subsection{Convergence}
\label{sec:convergence}

We will have the occasion to discuss the convergence of random variables. There are several notions of convergence in the literature, but for our purposes the most useful is \emph{convergence in probability}. Let $X$ and $X_n$, $n\ge 0$ be bounded nonnegative random variables. We say that the sequence $X_n$ \emph{converges to $X$ in probability} and write $X_n\cpr X$ if for all fixed $\delta>0$,
\begin{align*}
\Pr(\len{X_n-X} > \delta) = o(1).
\end{align*}

Let $\Exp X$ denote the expected value of $X$ and $\Var X$ its variance.
\begin{lemma}\ 
\label{lem:prfacts}
\begin{enumerate}[{\upshape(i)}]
\item\label{lem:prfactsi}
If $X_n\cpr X$ and $X_n\cpr Y$, then $X=Y$ with probability 1.
\item\label{lem:prfactsii}
If $X_n\cpr X$ and $Y_n\cpr Y$, then $X_n+Y_n\cpr X+Y$ and $X_nY_n\cpr XY$.
\item\label{lem:prfactsiii}
If $X_n\cpr X$ and $X$ is bounded away from 0, then $1/X_n\cpr 1/X$.
\item\label{lem:prfactsiv}
If $\Var{X_n}=o(1)$ and $\Exp{X_n}=e$ for all $n$, then $X_n\cpr e$.
\end{enumerate}
\end{lemma}
\begin{proof}
For \eqref{lem:prfactsiv}, by the Chebyshev bound $\Pr(\len{X-\Exp X} > k\sqrt{\Var X}) < 1/k^2$, for all fixed $\delta>0$,
\begin{align*}
\Pr(\len{X_n - e} > \delta) < \delta^{-2}\Var{X_n},
\end{align*}
and the \rhs\ is $o(1)$ by assumption.
\end{proof}

See \cite{Chung74,Feller71v1,Feller71v2,Doob53,Kolmogorov50,Kolmogorov56,durrett2010probability} for a more thorough introduction.

\subsection{Efficiency}
\label{sec:efficiency}

The \emph{efficiency} of a protocol is the long-term ratio of entropy production to entropy consumption. Formally, for a fixed protocol $\dd:S\times\Sigma\to S\times\Gs$, $s\in S$, and $\alpha\in\So$, define the random variables
\begin{align}
E_n(\alpha) = \frac{\len{\dd(s,\alpha_n)}}{n} \cdot \frac{H(\nu)}{H(\mu)},\label{eq:En}
\end{align}
where $H$ is the Shannon entropy
\begin{align*}
H(\seq p1n) &= -\sum_{i=1}^n p_i\log p_i
\end{align*}
(logarithms are base $2$ if not otherwise annotated), $\mu$ and $\nu$ are the input and output distributions, respectively, and $\alpha_n$ is the prefix of $\alpha$ of length $n$.
Intuitively, the Shannon entropy of a distribution measures the amount of randomness in it, where the basic unit of measurement is one fair coin flip. For example, as noted in the introduction, one roll of a fair six-sided die is worth about 2.585 coin flips. The random variable $E_n$ measures the ratio of entropy production to consumption after $n$ steps of $\dd$ starting in state $s$. Here $\len{\dd(s,\alpha_n)}\cdot H(\nu)$ (respectively, $n\cdot H(\mu)$) is the contribution along $\alpha$ to the production (respectively, consumption) of entropy in the first $n$ steps. We write $E_n^{\dd,s}$ when we need to distinguish the $E_n$ associated with different protocols and start states.

In most cases of interest, $E_n$ converges in probability to a unique constant value independent of start state and history. When this occurs, we call this constant value the \emph{efficiency} of the protocol $\dd$ and denote it by $\Eff\dd$. Notationally,
\begin{align*}
E_n \cpr \Eff\dd.
\end{align*}
One must be careful when analyzing infinite-state protocols: The efficiency is well-defined for finite-state protocols, but may not exist in general. For positive recurrent restart protocols, it is enough to measure the ratio for one iteration of the protocol.

In \S\ref{sec:serial} we will give sufficient conditions for the existence of $\Eff\dd$ that are satisfied by all protocols considered in \S\ref{sec:reduction}.

\subsection{Capacity}
\label{sec:capacity}

After reading some fixed number $n$ of random input symbols, the automaton implementing the protocol $\delta$ will have emitted a string of outputs $y_n$ and will also be in some random state $s_n$, where $(s_n,y_n)=\delta(s,a_1a_2\cdots a_n)$. The state $s_n$ will be distributed according to some distribution $\sigma_n$, which is induced by the distribution $\mu$ on inputs, therefore contains information $H(\sigma_n)$. We regard this quantity as information that is stored in the current state, later to be emitted as output or discarded. Any subsequent output entropy produced by the protocol is bounded by the sum of this stored entropy and additional entropy from further input.

Restart protocols operate by gradually consuming entropy from the input and storing it in the state, then emitting some fraction of the stored entropy as output all at once and returning to the start state. The stored entropy drops to $0$ at restart, reflecting the fact that no information is retained; any entropy that was not emitted as output is lost.

For finite-state protocols, the stored entropy is bounded by the base-2 logarithm of the size of the state space, the entropy of the uniform distribution. We call this quantity the \emph{capacity} of the protocol:
\begin{align}
\Cp\delta &= \log_2\len S.\label{def:capacity}
\end{align}
The capacity is a natural measure of the complexity of $\delta$, and we will take it as our complexity measure for finite-state protocols. In \S\ref{sec:reduction}, we will construct families of protocols for various reductions indexed by a tunable parameter $k$ proportional to the capacity. The efficiency of the protocols is expressed as a function of $k$; by choosing larger $k$, greater efficiency can be achieved at the cost of a larger state space. The results of \S\ref{sec:reduction} quantify this tradeoff.

\subsection{Entropy and Conditional Entropy}
\label{sec:entropy}

In this subsection we review a few elementary facts about entropy and conditional entropy that we will need. These are well known; the reader is referred to \cite{CoverThomas91,Adamek91} for a more thorough treatment.

Let $p = (p_n : n\in N)$ be any discrete finite or countably infinite subprobability distribution (that is, all $p_n\ge 0$ and $\sum_{n\in N}p_n\le 1$) with finite entropy
\begin{align*}
H(p) = H(p_n : n\in N) &= -\sum_{n\in N} p_n\log p_n < \infty.
\end{align*}
For $E\subs N$, define $p_E = \sum_{n\in E} p_n$. The \emph{conditional entropy} with respect to the event $E$ is defined as
\begin{align}
H(p \mid E) &= H(\frac{p_n}{p_E} : n\in E)
= -\sum_{n\in E} \frac{p_n}{p_E}\log \frac{p_n}{p_E}.\label{eq:condentdef}
\end{align}
It follows that
\begin{align}
H(p_n : n\in E) &= p_E H(p \mid E) - p_E\log p_E.\label{eq:ent1}
\end{align}
A \emph{partition} of $N$ is any finite or countable collection of nonempty pairwise disjoint subsets of $N$ whose union is $N$.
\begin{lemma}[Conditional entropy rule; see {\cite[\S2.2]{CoverThomas91}}]
\label{lem:condent}
Let $p = (p_n : n\in N)$ be a discrete subprobability distribution with finite entropy, and let $\EE$ be any partition of $N$. Then
\begin{align*}
H(p) = H(p_A : A\in\EE) + \sum_{A\in\EE} p_A H(p \mid A).
\end{align*}
\end{lemma}
\begin{proof}
From \eqref{eq:ent1},
\begin{align*}
\sum_{A\in\EE} p_A H(p \mid A)
&= \sum_{A\in\EE} H(p_n : n\in A) + \sum_{A\in\EE} p_A \log p_A\\
&= H(p) - H(p_A : A\in\EE).
\qedhere
\end{align*}
\end{proof}

It is well known that the probability distribution on $d$ letters that maximizes entropy is the uniform distribution with entropy $\log d$ (see \cite{CoverThomas91}). A version of this is also true for subprobability distributions:
\begin{lemma}
\label{lem:ent3}
The uniform subprobability distribution $(s/d,\ldots,s/d)$ on $d$ letters with total mass $s$ and entropy $s \log(d/s)$ maximizes entropy among all subprobability distributions on $d$ letters with total mass $s$.
\end{lemma}
\begin{proof}
For any subprobability distribution $(\seq p1d)$ with $s=\sum_{i=1}^d p_i$, it follows from the definitions that
\begin{align*}
H(\seq p1d)
&= sH(\frac{p_1}s,\ldots,\frac{p_d}s) - s\log s\\
&\le sH(\frac 1d,\ldots,\frac 1d) - s\log s
= H(\frac sd,\ldots,\frac sd)
= s \log \frac ds.
\qedhere
\end{align*}
\end{proof}


\section{Basic Results}
\label{sec:basic}

Let $\dd:S\times\Sigma\to S\times\Gs$ be a protocol reducing $\mu$ to $\nu$.
We can associate with each $y\in\Gs$ and state $s\in S$ a prefix code in $\Ss$, namely
\begin{align}
\pcd y &= \{\text{$\prec$-minimal strings $x\in\Ss$ such that $y\prefeq\dd(s,x)$}\}.\label{eq:pc}
\end{align}
The string $y$ is generated as a prefix of the output if and only if exactly one $x\in\pcd y$ is consumed as a prefix of the input. These events must occur with the same probability, so
\begin{align}
\nu(y) = \Pr(y \pref \ds(s,\alpha)) = \mu(\pcd y),\label{eq:munu}
\end{align}
where $\mu(\pcd y)$ is defined in \eqref{eq:prefixcode}.
Note that $\pcd y$ need not be finite.

\begin{lemma}
If $A\subs\Gs$ is a prefix code, then so is $\bigcup_{y\in A}\pcd y\subs\Ss$, and
\begin{align*}
\nu(A) = \mu(\bigcup_{y\in A}\pcd y).
\end{align*}
If $A\subs\Gs$ is exhaustive, then so is $\bigcup_{y\in A}\pcd y\subs\Ss$.
\end{lemma}
\begin{proof}
We have observed that each $\pcd y$ is a prefix code. If $y_1$ and $y_2$ are $\prefeq$-incomparable, and if $y_1\prefeq\dd(s,x_1)$
and $y_2\prefeq\dd(s,x_2)$, then $x_1$ and $x_2$ are $\prefeq$-incomparable, thus $\bigcup_{y\in A}\pcd y$ is a prefix code. By \eqref{eq:munu}, we have
\begin{align*}
\nu(A) = \sum_{y\in A}\nu(y) = \sum_{y\in A}\mu(\pcd y) = \mu(\bigcup_{y\in A}\pcd y).
\end{align*}
If $A\subs\Gs$ is exhaustive, then so is $\bigcup_{y\in A}\pcd y$, since the events both occur with probability 1 in their respective spaces.
\end{proof}

\begin{lemma}\ 
\label{lem:delta}
\begin{enumerate}[{\upshape(i)}]
\item\label{lem:deltai}
The partial function $\ds(s,-):\So\pfun\Go$ is continuous, thus Borel measurable.
\item\label{lem:deltaii}
$\ds(s,\alpha)$ is almost surely infinite; that is, $\mu(\dom\ds(s,-)) = 1$.
\item\label{lem:deltaiii}
The measure $\nu$ on $\Go$ is the push-forward measure $\nu = \mu\circ\ds(s,-)^{-1}$.
\end{enumerate}
\end{lemma}
\begin{proof}
\eqref{lem:deltai}
Let $y\in\Gs$.
The preimage of $\set{\beta\in\Go}{y\pref\beta}$, a basic open set of $\Go$, is open in $\So$:
\begin{align*}
\ds(s,-)^{-1}(\set{\beta}{y\pref\beta}) &= \set{\alpha}{y\pref\ds(s,\alpha)}
= \bigcup_{x\in\pcd y}\,\set{\alpha}{x\pref\alpha}.
\end{align*}

\eqref{lem:deltaii}
We have assumed finite latency; that is, starting from any state, the expected time before the next output symbol is generated is finite. Thus the probability that infinitely many symbols are generated is 1.

\eqref{lem:deltaiii}
From \eqref{lem:deltai} and \eqref{eq:munu} we have
\begin{align*}
(\mu\circ\ds(s,-)^{-1})(\set{\beta}{y\pref\beta})
&= \mu(\bigcup_{x\in \pcd y} \set\alpha{x\pref\alpha})\\
&= \mu(\pcd y)
= \nu(y)
= \nu(\set{\beta}{y\pref\beta}).
\end{align*}
Since $\mu\circ\ds(s,-)^{-1}$ and $\nu$ agree on the basic open sets $\set{\beta}{y\pref\beta}$, they are equal.
\end{proof}

\begin{lemma}\ 
\label{lem:absolute}
If $\dd$ is a reduction from $\mu$ to $\nu$, then the random variables $E_n$ defined in \eqref{eq:En} are continuous and uniformly bounded by an absolute constant $R>0$ depending only on $\mu$ and $\nu$.
\end{lemma}
\begin{proof}
For $x\in\Ss$, let $y$ be the string of output symbols produced after consuming $x$. The protocol cannot produce $y$ from $x$ with greater probability than allowed by $\nu$, thus
\begin{align*}
(\min_{a\in\Sigma} \mu(a))^{\len x} \le \mu(x) \le \nu(y) \le (\max_{b\in\Gamma} \nu(b))^{\len y}.
\end{align*}
Taking logs, $\len y \le \len x\log\min_{a\in\Sigma}\mu(a)/\log\max_{b\in\Gamma}\nu(b)$, thus we can choose
\begin{align*}
R=\frac{H(\nu)\log\min_{a\in\Sigma}\mu(a)}{H(\mu)\log\max_{b\in\Gamma}\nu(b)}.
\end{align*}

To show continuity, for $r\in\reals$,
\begin{align*}
E_n^{-1}(\set x{x < r})
&= \set{\alpha}{\len{\dd(s,\alpha_n)} < nrH(\mu)/H(\nu)}\\
&= \bigcup\,\set{\set{\alpha}{x\pref\alpha}}{\len x=n,\ \len{\dd(s,x)} < nrH(\mu)/H(\nu)},
\end{align*}
an open set.
\end{proof}

\subsection{Composition}
\label{sec:composition}

Protocols can be composed sequentially as follows. If
\begin{align*}
& \dd_1:S\times\Sigma\to S\times\Gs && \dd_2:T\times\Gamma\to T\times\Ds,
\end{align*}
then
\begin{align*}
& (\dd_1\cmp\dd_2):S\times T\times\Sigma\to S\times T\times\Ds\\
& (\dd_1\cmp\dd_2)((s,t),a) = \letin{(u,y)}{\dd_1(s,a)}{\letin{(v,z)}{\dd_2(t,y)}{((u,v),z)}}.
\end{align*}
Intuitively, we run $\dd_1$ for one step and then run $\dd_2$ on the output of $\dd_1$.
The following theorem shows that the partial map on infinite strings induced by the sequential composition of protocols agrees almost everywhere with the functional composition of the induced maps of the component protocols. 
\begin{theorem}
\label{thm:composition}
The partial map $\ds_2(t,\ds_1(s,-))$ of type $\Sigma^\omega\pfun\Delta^\omega$ is defined on all but a $\mu$-nullset and agrees with $(\dd_1\cmp\dd_2)^\omega((s,t),-)$ on its domain of definition.
\end{theorem}
\begin{proof}
We restrict inputs to the subset of $\So$ on which $\ds_1(s,-)$ is defined and produces a string in $\Go$ on which $\ds_2(t,-)$ is defined. This set is of measure 1: if $\ds_1$ reduces $\mu$ to $\nu$ and $\ds_2$ reduces $\nu$ to $\rho$, then by Lemma \ref{lem:delta}\eqref{lem:deltaiii}, 
\begin{align*}
\mu(\dom\ds_2(t,\ds_1(s,-))) &= \mu(\ds_1(s,-)^{-1}(\ds_2(t,-)^{-1}(\Delta^\omega)))\\
&= \nu(\ds_2(t,-)^{-1}(\Delta^\omega))
= \rho(\Delta^\omega)
= 1.
\end{align*}
Thus we only need to show that
\begin{align}
(\dd_1\cmp\dd_2)^\omega((s,t),\alpha) = \ds_2(t,\ds_1(s,\alpha))\label{eq:deltadelta}
\end{align}
for inputs $\alpha$ in this set.

We show \eqref{eq:deltadelta} by coinduction. A \emph{bisimulation} on infinite streams is a binary relation $R$ such that if $R(\beta,\gamma)$, then there exists a finite nonnull string $z$ such that
\begin{align}
\beta &= z\beta' & \gamma &= z\gamma' & R(\beta',\gamma').\label{eq:bisim}
\end{align}
That is, $\beta$ and $\gamma$ agree on a finite nonnull prefix $z$, and deleting $z$ from the front of $\beta$ and $\gamma$ preserves membership in the relation $R$. The coinduction principle on infinite streams says that if there exists a bisimulation $R$ such that $R(\beta,\gamma)$, then $\beta=\gamma$.

We will apply this principle with the binary relation
\begin{align*}
R(\beta,\gamma)\ &\Iff\ \exists \alpha\in\Sigma^\omega\ \exists s\in S\ \exists t\in T\ \ \beta=(\dd_1\cmp\dd_2)^\omega((s,t),\alpha) \wedge \gamma=\ds_2(t,\ds_1(s,\alpha)) 
\end{align*}
on $\Delta^\omega$. To show that this is a bisimulation, suppose $R(\beta,\gamma)$ with
\begin{align*}
\beta &= (\dd_1\cmp\dd_2)^\omega((s,t),a\alpha) & \gamma &= \ds_2(t,\ds_1(s,a\alpha)),
\end{align*}
where $a\in\Sigma$ and $\alpha\in\So$.
Unwinding the definitions,
\begin{align*}
\beta &= (\dd_1\cmp\dd_2)^\omega((s,t),a\alpha)\\
&= \letin{((u,v),z)}{(\dd_1\cmp\dd_2)((s,t),a)}{z\cdot(\dd_1\cmp\dd_2)^\omega((u,v),\alpha)}\\
&= \letin{(u,y)}{\dd_1(s,a)}{\letin{(v,z)}{\dd_2(t,y)}{z\cdot(\dd_1\cmp\dd_2)^\omega((u,v),\alpha)}}\\
\gamma &= \ds_2(t,\ds_1(s,a\alpha))\\
&= \letin{(u,y)}{\dd_1(s,a)}{\letin{\zeta}{\ds_1(u,\alpha)}{\ds_2(t,y\zeta)}}\\
&= \letin{(u,y)}{\dd_1(s,a)}{\letin{\zeta}{\ds_1(u,\alpha)}{\letin{(v,z)}{\dd_2(t,y)}{z\cdot\ds_2(v,\zeta)}}}\\
&= \letin{(u,y)}{\dd_1(s,a)}{\letin{(v,z)}{\dd_2(t,y)}{z\cdot\ds_2(v,\ds_1(u,\alpha))}},
\end{align*}
so if $(u,y)=\dd_1(s,a)$ and $(v,z)=\dd_2(t,y)$, then
\begin{align*}
\beta &= (\dd_1\cmp\dd_2)^\omega((s,t),a\alpha)
= z\cdot(\dd_1\cmp\dd_2)^\omega((u,v),\alpha) = z\beta'\\
\gamma &= \ds_2(t,\ds_1(s,a\alpha))
= z\cdot\ds_2(v,\ds_1(u,\alpha)) = z\gamma',
\end{align*}
where
\begin{align*}
\beta' &= (\dd_1\cmp\dd_2)^\omega((s,t),\alpha) & \gamma' &= \ds_2(t,\ds_1(s,\alpha)) & R(\beta',\gamma').
\end{align*}
We almost have \eqref{eq:bisim}, except that $z$ may be the null string, in which case $\beta=\beta'$ and $\gamma=\gamma'$, and we cannot conclude yet that $R$ is a bisimulation. But in this case we unwind again in the same way, and continue to unwind until we get a nonnull $z$, which must happen after finitely many steps by Lemma \ref{lem:delta}\eqref{lem:deltaii}. Thus $R$ is a bisimulation.

By the principle of coinduction, we can conclude \eqref{eq:deltadelta} for all $\alpha$ in the domain of definition of $\ds_2(t,\ds_1(s,-))$.
\end{proof}

\begin{corollary}
\label{cor:composition1}
If $\dd_1(s,-)$ is a reduction from $\mu$ to $\nu$ and $\dd_2(t,-)$ is a reduction from $\nu$ to $\rho$, then $(\dd_1\cmp\dd_2)((s,t),-)$ is a reduction from $\mu$ to $\rho$.
\end{corollary}
\begin{proof}
By the assumptions in the statement of the corollary,
$\nu = \mu\circ\ds_1(s,-)^{-1}$ and $\rho = \nu\circ\ds_2(t,-)^{-1}$.
By Theorem \ref{thm:composition},
\begin{align*}
\rho &= \mu\circ\ds_1(s,-)^{-1}\circ\ds_2(t,-)^{-1}
= \mu\circ(\ds_2(t,-)\circ\ds_1(s,-))^{-1}\\
&= \mu\circ(\ds_2(t,\ds_1(s,-)))^{-1}
= \mu\circ((\dd_1\cmp\dd_2)^\omega((s,t),-))^{-1}.
\qedhere
\end{align*}
\end{proof}

\begin{theorem}
\label{thm:composition2}
If $\dd_1(s,-)$ is a reduction from $\mu$ to $\nu$ and $\dd_2(t,-)$ is a reduction from $\nu$ to $\rho$, and if $\Eff{\dd_1}$ and $\Eff{\dd_2}$ exist, then $\Eff{\dd_1\cmp\dd_2}$ exists and
\begin{align*}
\Eff{\dd_1\cmp\dd_2} &= \Eff{\dd_1}\cdot\Eff{\dd_2}. 
\end{align*}
\end{theorem}
\begin{proof}
Let $\alpha\in\dom(\dd_1\cmp\dd_2)^\omega((s,t),-)$, say $\ds_1(s,\alpha)=\beta\in\Go$ with $\beta\in\dom\ds_2(t,-)$. Let $n\in\naturals$. The second component of $\dd_1(s,\alpha_n)$ is $\beta_m$ for some $m$, and $\len{\beta_m}=m=\len{\dd_1(s,\alpha_n)}$. Then
\begin{align*}
\frac{\len{(\dd_1\cmp\dd_2)((s,t),\alpha_n)}}{n}\cdot\frac{H(\rho)}{H(\mu)}
&= \frac{\len{\dd_2(t,\beta_m)}}{n}\cdot\frac{H(\rho)}{H(\mu)}\\
&= \frac{\len{\dd_2(t,\beta_m)}}{\len{\beta_m}}\cdot\frac{\len{\beta_m}}{n}\cdot\frac{H(\rho)}{H(\nu)}\cdot\frac{H(\nu)}{H(\mu)}\\
&= (\frac{\len{\dd_1(s,\alpha_n)}}{n}\cdot\frac{H(\nu)}{H(\mu)})
(\frac{\len{\dd_2(t,\beta_m)}}{m}\cdot\frac{H(\rho)}{H(\nu)})\\
&= E_n^{\dd_1,s}(\alpha)\cdot E_m^{\dd_2,t}(\beta).
\end{align*}
By Lemma \ref{lem:prfacts}\eqref{lem:prfactsii}, this quantity converges in probability to $\Eff{\dd_1}\cdot\Eff{\dd_2}$, so this becomes $\Eff{\dd_1\cmp\dd_2}$. 
\end{proof}

The capacity of the composition is additive:

\begin{theorem}
\label{thm:compositioncap}
For finite-state protocols, $\Cp{\dd_1\cmp\dd_2} = \Cp{\dd_1} + \Cp{\dd_2}$.
\end{theorem}
\begin{proof}
Immediate from the definition.
\end{proof}

\subsection{Protocol Families}
\label{sec:families}

In \S\ref{sec:reduction}, we will present families of reductions between concrete pairs of distributions. The families are indexed by a parameter $k$, which controls the tradeoff between the capacity of the protocol, proportional to $k$,
and its efficiency, typically expressed in the form $1-\Theta(f(k))$. Higher efficiency comes at the cost of higher capacity. Asymptotically optimal reductions were known to exist for all finite distributions (\cite{vonNeumann51,Elias72,MosselPeresHillar05,Peres92,PaeLoui06,PaeLoui05,Pae05}, cf.~Theorem \ref{thm:opt}); however, by considering the rate of convergence as a function of $k$, we obtain a natural measure of quality for a family of protocols that allows a finer-grained comparison.

A key consequence of our composition theorems (Theorems \ref{thm:composition2} and \ref{thm:compositioncap}) is that this notion of quality is preserved under composition. To make this notion precise, we first formalize protocol families.
\begin{definition}
We say that a sequence $\mathcal{P}=(P_k : k\in\naturals)$ is a \emph{capacity-indexed family of reductions} (\emph{cfr}) from $\mu$ to $\nu$ if
\begin{itemize}
\item each $P_k$ is a reduction from $\mu$ to $\nu$;
\item each $P_k$ has capacity $\Theta(k)$, that is, there exist constants $c_1$, $c_2$ independent of $k$ such that $c_1k\leq \Cp{P_k} \leq c_2k$.
\end{itemize}
\end{definition}
The notion of efficiency of a single reduction naturally generalizes to capacity-indexed families, as we can take
the efficiency $\Eff\mathcal{P}$ of a family $\mathcal{P}$ to be the function from the index $k$ to the efficiency of the $k$th protocol.

\begin{theorem}\label{thm:familycomposition}
Suppose $\mathcal{P}=(P_k : k\in\naturals)$ is a cfr from $\mu$ to $\nu$ with
efficiency $1-f(k)$ and $\mathcal{Q}=(Q_k : k\in\naturals)$ is a cfr from $\nu$ to $\rho$ with efficiency $1-g(k)$,
where $f(k)$ and $g(k)$ are non-negative real-valued functions.
Then $\mathcal{P}\cmp\mathcal{Q}=(P_k\cmp Q_k : k\in\naturals)$ is a cfr, and its efficiency is $\Eff(\mathcal{P}\cmp\mathcal{Q})(k)=(1-f(k))(1-g(k))\ge 1-(f(k)+g(k))$.
\end{theorem}
\begin{proof}
By Corollary \ref{cor:composition1}, each component of $\mathcal{P}\cmp\mathcal{Q}$ is a reduction from $\mu$ to $\rho$, and by Theorem \ref{thm:compositioncap}, its capacity is again $\Theta(k)$. That the efficiency of the composition exists and satisfies the stated bounds follows immediately from Theorem \ref{thm:composition2}.
\end{proof}

In other words, protocol families can be composed, and the resulting protocol family is asymptotically no worse than the worst of the two input families.

\begin{example}
In Section \ref{sec:dtorat}, we construct a cfr from
$c$-uniform to arbitrary rational distributions with efficiency $1-\Theta(k^{-1})$,
and in Section \ref{sec:arb2unif}, we construct a cfr from
an arbitrary distributions to a $c$-uniform one with
efficiency $1-\Theta(\log k/k)$. These two families can be
composed to obtain a cfr from an arbitrary distribution
to an arbitrary rational distribution. By Theorem \ref{thm:familycomposition},
the resulting cfr has efficiency $1-\Theta(\log k/k)$.
\end{example}

\subsection{Serial Protocols}
\label{sec:serial}

Consider an infinite sequence $(S_0,\dd_0,s_0)$, $(S_1,\dd_1,s_1),\ldots$ of positive recurrent restart protocols defined in terms of maps $f_k:A_k\to\Gs$, where the $A_k$ are exhaustive prefix codes, as described in \S\ref{sec:restart}. These protocols can be combined into a single \emph{serial protocol} $\dd$. Intuitively, the serial protocol starts in $s_0$, makes $\delta_0$-steps in $S_0$ accumulating the consumed letters until the consumed string $x$ is in $A_0$, then produces $f_0(x)$ and transitions to $s_1$, where it then repeats these steps for protocol $(S_1,\delta_1,s_1)$, then for $(S_2,\delta_2,s_2)$, and so on. Formally, the states of $\dd$ are the disjoint union of the $S_k$, and $\dd$ is defined so that $\dd(s_k,x) = (s_{k+1},f_k(x))$ for $x\in A_k$, and within $S_k$ behaves like $\dd_k$.

Let $C_k$ be a random variable representing the entropy consumption of the component protocol $\dd_k$ starting from $s_k$ during the execution of the serial protocol; that is, $C_k$ is the number of input symbols consumed by $\dd_k$ scaled by $H(\mu)$. This is a random variable whose values depend on the input sequence $\alpha\in\So$. Note that $C_k$ may be partial, but is defined with probability one by the assumption of bounded latency. Similarly, let $P_k$ be the number of output symbols written during the execution of $\dd_k$ scaled by $H(\nu)$. Let $e(n)$ be the index of the component protocol $\dd_{e(n)}$ in which the $n$-th step of the combined protocol occurs. Like $C_k$, $P_k$ and $e(n)$ are random variables whose values depend on the input sequence $\alpha\in\So$. Let $c_k=\Exp{C_k}$ and $p_k=\Exp{P_k}$.

To derive the efficiency of serial protocols, we need a form of the law of large numbers (see \cite{Chung74,Feller71v1}). Unfortunately, the law of large numbers as usually formulated does not apply verbatim, as the random variables in question are bounded but not independent, or (under a different formulation) independent but not bounded. Our main result, Theorem \ref{thm:opt} below, can be regarded as a specialized version of this result adapted to our needs.

Our version requires that the variances of certain random variables vanish in the limit. We need to impose mild conditions \eqref{eq:cassm} on the growth rate of $m_n$, the maximum consumption in the $n$th component protocol, and the growth rate of production relative to consumption. These conditions hold for all serial protocols considered in this paper. The left-hand condition of \eqref{eq:cassm} is satisfied by all serial protocols in which either $m_n$ is bounded or $m_n=O(n)$ and $\liminf_n c_n=\infty$.

\begin{lemma}
\label{lem:variance}
Let $\Var X$ denote the variance of $X$.
Let $m_n = \max_{x\in A_n}\len x\cdot H(\mu)$ and suppose that $m_n$ is finite for all $n$. If
\begin{align}
& m_n = o(\sum_{i=0}^{n-1} c_i) && \sum_{i=0}^np_i=\Omega(\sum_{i=0}^nc_i),\label{eq:cassm}
\end{align}
then
\begin{align}
\Var{\frac{\sum_{i=0}^n C_i}{\sum_{i=0}^n c_i}} &= o(1) &
\Var{\frac{C_n}{\sum_{i=0}^{n-1} c_i}} &= o(1)\label{eq:variance}\\[1ex]
\Var{\frac{\sum_{i=0}^n P_i}{\sum_{i=0}^n p_i}} &= o(1) &
\Var{\frac{P_n}{\sum_{i=0}^{n-1} p_i}} &= o(1).\label{eq:variance1}
\end{align}
\end{lemma}
\begin{proof}
The properties \eqref{eq:variance} require only the left-hand condition of \eqref{eq:cassm}. Let $\eps>0$ be arbitrarily small. Choose $m$ such that $m_i/\sum_{j<i} c_j < \eps$ for all $i\ge m$, then choose $n>m$ such that $m_i/\sum_{j=0}^n c_j < \eps$ for all $i<m$. As the $C_i$ are independent,
\begin{align*}
\Var{\frac{\sum_{i=0}^n C_i}{\sum_{i=0}^n c_i}}
&= \sum_{i=0}^n \frac{\Var{C_i}}{(\sum_{j=0}^n c_j)^2}
\le \sum_{i=0}^n \frac{\Exp{C_i^2}}{(\sum_{j=0}^n c_j)^2}\\
&= \sum_{i=0}^{m-1}\Exp{\frac{C_i}{\sum_{j=0}^n c_j}\cdot\frac{C_i}{\sum_{j=0}^n c_j}} + \sum_{i=m}^{n}\Exp{\frac{C_i}{\sum_{j=0}^n c_j}\cdot\frac{C_i}{\sum_{j=0}^n c_j}}\\
&\le \sum_{i=0}^{m-1}\Exp{\frac{m_i}{\sum_{j=0}^n c_j}\cdot\frac{C_i}{\sum_{j=0}^n c_j}} + \sum_{i=m}^{n}\Exp{\frac{m_i}{\sum_{j=0}^{i-1} c_j}\cdot\frac{C_i}{\sum_{j=0}^n c_j}}\\
&\le \sum_{i=0}^{m-1}\Exp{\frac{\eps C_i}{\sum_{j=0}^n c_j}} + \sum_{i=m}^{n}\Exp{\frac{\eps C_i}{\sum_{j=0}^n c_j}}\ =\ \sum_{i=0}^{n}\frac{\eps c_i}{\sum_{j=0}^n c_j}\ =\ \eps\\
\Var{\frac{C_n}{\sum_{i=0}^{n-1} c_i}}
&\le \frac{\Exp{C_n^2}}{(\sum_{j=0}^{n-1} c_j)^2}
\le \frac{m_n^2}{(\sum_{j=0}^{n-1} c_j)^2} \le \eps^2.
\end{align*}
As $\eps$ was arbitrarily small, \eqref{eq:variance} holds.

If in addition the right-hand condition of \eqref{eq:cassm} holds, then by Lemma \ref{lem:absolute}, $m_n = o(\sum_{i=0}^{n-1} p_i)$ for all $n$. Then \eqref{eq:variance1} follows by the same proof with $P_i$, $p_i$, and $Rm_i$ substituted for $C_i$, $c_i$, and $m_i$, respectively.
\end{proof}

The following theorem, in conjunction with the constructions of \S\ref{sec:reduction}, shows that optimal efficiency is achievable in the limit. The result is mainly of theoretical interest, since the protocols involve infinitely many states.

\begin{theorem}
\label{thm:opt}
Let $\dd$ be a serial protocol with finite-state components $\dd_0,\dd_1,\ldots$ satisfying \eqref{eq:cassm}. If the limit
\begin{align}
\ell &= \lim_n \frac{\sum_{i=0}^n p_i}{\sum_{i=0}^n c_i}\label{eq:opt2}
\end{align}
exists, then the efficiency of the serial protocol exists and is equal to $\ell$.
\end{theorem}
\begin{proof}
The expected time in each component protocol is finite, thus $e(n)$ is unbounded with probability 1. By definition of $e(n)$, we have
\begin{align*}
\sum_{i=0}^{e(n)-1}C_{i} &\le n\cdot H(\mu) \le \sum_{i=0}^{e(n)}C_{i}
&
\sum_{i=0}^{e(n)-1}P_{i} &\le \len{\delta(s,\alpha_n)}\cdot H(\nu) \le \sum_{i=0}^{e(n)}P_{i},
\end{align*}
therefore
\begin{align}
\frac{\sum_{i=0}^{e(n)-1}P_{i}}{\sum_{i=0}^{e(n)}C_{i}}
&\le \frac{\len{\delta(s,\alpha_n)}}{n}\cdot \frac{H(\nu)}{H(\mu)} = E_n(\alpha)
\le \frac{\sum_{i=0}^{e(n)}P_{i}}{\sum_{i=0}^{e(n)-1}C_{i}}.\label{eq:El}
\end{align}
By Lemma \ref{lem:variance}, the variance conditions \eqref{eq:variance} and \eqref{eq:variance1} hold. Then by Lemma \ref{lem:prfacts}\eqref{lem:prfactsiv},
\begin{align*}
\frac{\sum_{i=0}^n C_i}{\sum_{i=0}^n c_i} &\cpr 1 &
\frac{\sum_{i=0}^n P_i}{\sum_{i=0}^n p_i} &\cpr 1 &
\frac{C_n}{\sum_{i=0}^{n-1} c_i} &\cpr 0 &
\frac{P_n}{\sum_{i=0}^{n-1} p_i} &\cpr 0.
\end{align*}
Using Lemma \ref{lem:prfacts}\eqref{lem:prfactsi}-\eqref{lem:prfactsiii}, we have
\begin{align*}
\frac{\sum_{i=0}^{n} P_i}{\sum_{i=0}^{n-1} C_i}
&=
(\frac{P_n}{\sum_{i=0}^{n-1} p_i} + \frac{\sum_{i=0}^{n-1} P_i}{\sum_{i=0}^{n-1} p_i})\cdot
\frac{\sum_{i=0}^{n-1} p_i}{\sum_{i=0}^{n-1} c_i}\cdot
\frac{\sum_{i=0}^{n-1} c_i}{\sum_{i=0}^{n-1} C_i} \cpr \ell
\end{align*}
and similarly $\sum_{i=0}^{n-1} P_i/\sum_{i=0}^{n} C_i\cpr\ell$. The conclusion $E_n \cpr \ell$ now follows from \eqref{eq:El}.
\end{proof}


\section{Reductions}
\label{sec:reduction}

In this section we present a series of reductions between distributions of certain forms. Each example defines a capacity-indexed family of reductions (\S\ref{sec:families}) given as positive recurrent restart protocols (\S\ref{sec:restart})
with efficiency tending to 1 as the parameter $k$ grows. 
By Theorem \ref{thm:opt}, each family can be made into a single serial protocol (\S\ref{sec:serial}) with asymptotically optimal efficiency, and by Theorem \ref{thm:familycomposition}, any two compatible reduction families with asymptotically optimal efficiency can be composed to form a family of reductions with asymptotically optimal efficiency. 

\subsection{Uniform $\Imp$ Uniform}
\label{sec:unif}

Let $c,d\ge 2$, the sizes of the output and input alphabets, respectively. In this section we construct a family of restart protocols with capacity proportional to $k$ mapping $d$-uniform streams to $c$-uniform streams with efficiency $1-\Theta(k^{-1})$. The Shannon entropy of the input and output distributions are $\log d$ and $\log c$, respectively.

Let $m=\floor{k\log_c d}$. Then $c^m \le d^k < c^{m+1}$. It follows that
\begin{align}
\frac 1c < \frac{c^m}{d^k} &\le 1 & 1-\frac{\log c}{k\log d} &< \frac{m\log c}{k\log d} \le 1.\label{eq:a1}
\end{align}
Let the $c$-ary expansion of $d^k$ be
\begin{align}
d^k &= \sum_{i=0}^{m} a_i c^i,\label{eq:a2}
\end{align}
where $0\le a_i \le c-1$, $a_m\neq 0$.

Intuitively, the protocol $P_k$ operates as follows. Do $k$ calls on the $d$-uniform distribution. For each $0\le i\le m$, for $a_ic^i$ of the possible outcomes, emit a $c$-ary string of length $i$, every possible such string occurring exactly $a_i$ times. For $a_0$ outcomes, nothing is emitted, and this is lost entropy, but this occurs with probability $a_0d^{-k}$. After that, restart the protocol.

Formally, this is a restart protocol with prefix code $A$ consisting of all $d$-ary strings of length $k$. For each of the $d^k$ strings $x\in A$, we specify an output string $f(x)$ to emit. Partition $A$ into $\sum_{i=0}^{m} c^i$ disjoint sets $A_{iy}$, one for each $0\le i\le m$ and $c$-ary string $y$ of length $i$, such that $\len{A_{iy}}=a_i$. The total number of strings in all partition elements is given by \eqref{eq:a2}. Set $f(x)=y$ for all $x\in A_{iy}$.

By elementary combinatorics,
\begin{align}
\sum_{i=0}^{m-1} (m - i)a_i c^i 
&\le \sum_{i=0}^{m-1} (m-i)(c-1)c^i = \frac{c(c^m-1)}{c-1} - m \le \frac{c(d^k-1)}{c-1} - m.\label{eq:a3}
\end{align}
In each run of $P_k$, the expected number of $c$-ary digits produced is
\begin{align}
\sum_{i=0}^{m} i a_i c^id^{-k}
&= d^{-k}(\sum_{i=0}^m ma_i c^i - \sum_{i=0}^{m} (m - i) a_i c^i)\nonumber\\
&\ge m - d^{-k}(\frac{c(d^k-1)}{c-1} - m) && \text{by \eqref{eq:a2} and \eqref{eq:a3}}\nonumber\\
&= m(1 + d^{-k}) - \frac c{c-1} (1-d^{-k})\nonumber\\
&\ge m - \frac c{c-1},\label{eq:a4}
\end{align}
thus the entropy production is at least $m\log c - \Theta(1)$.
The number of $d$-ary digits consumed is $k$, thus the entropy consumption is $k\log d$, which is also the
capacity. By \eqref{eq:a1} and \eqref{eq:a4}, the efficiency is at least
\begin{align*}
\frac{(m - \frac c{c-1})\log c}{k\log d} \ge 1 - \Theta(k^{-1}).
\end{align*}

The output is uniformly distributed, as there are $\sum_{i=\ell}^{m} a_i c^i$ equal-probability outcomes that produce a string of length $\ell$ or greater, and each output letter $a$ appears as the $\ell$th output letter in equally many strings of the same length, thus is output with equal probability.

\begin{example}
For $d=3$ and $c=k=2$, the prefix code $A$ would contain the nine ternary strings of length two. The binary expansion of $9$ is $1001$, which indicates that $\set{f(x)}{x\in A}$ should contain the eight binary strings of length three and the null string. 
The expected number of binary digits produced is $8/9\cdot 3 + 1/9\cdot 0 = 8/3$ and the expected number of ternary digits consumed is 2, so the production entropy is 8/3 and the consumption entropy is $2\log 3$ for an efficiency of about $.841$.
\end{example}


\subsection{Uniform $\Imp$ Rational}
\label{sec:dtorat}

Let $c,d\ge 2$. In this section, we will present a family of restart protocols $D_k$ mapping $d$-uniform streams over $\Sigma$ 
to streams over a $c$-symbol alphabet $\Gamma=\{1,\ldots,c\}$ with rational probabilities with a common denominator $e$, that is, $p_i = a_i/e$ for $i\in\Gamma$. By composing with a protocol of \S\ref{sec:unif} if necessary, we can assume without loss of generality that $e=d$, thus we assume that $p_i = a_i/d$ for $i\in\Gamma$.

Unlike the protocols in the previous section, here we emit a fixed number $k$ of symbols in each round while consuming a variable number of input symbols according to a particular prefix code $S\subs\Ss$.
The protocol $D_k$ will have capacity at most $k\log d$ and efficiency $1-\Theta(k^{-1})$, exhibiting a similar 
tradeoff to the family of \S\ref{sec:unif}.

To define $D_k$, we will construct a finite exhaustive prefix code $S$ over the source alphabet. The codewords of this prefix code will be partitioned into pairwise disjoint nonempty sets $S_y\subs\Ss$ associated with each $k$-symbol output word $y \in \Gamma^k$. All input strings in the set $S_y$ will map to the output string $y$.

Intuitively, the protocol operates as follows. Starting in the start state $s$, it reads input symbols until it has read an entire codeword, which must happen eventually since the code is exhaustive. If that codeword is in $S_y$, it emits $y$ and restarts. An example is given at the end of this section.

Let $p_y$ denote the probability of the word $y=e_1\cdots e_k$ in the output process, where $e_i\in\Gamma$, $1\le i\le k$. Since the symbols $e_i$ are chosen independently, $p_y$ is the product of the probabilities of the individual symbols. It is therefore of the form $p_y = a_y d^{-k}$, where $a_y=a_{e_1}\cdots a_{e_k}$ is an integer.

Let $m_y = \floor{\log_d{a_y}}$ and let
\begin{align*}
a_y &= \sum_{j=0}^{m_y} a_{y j}d^j
\end{align*}
be the $d$-ary expansion of $a_y$. We will choose a set of $\sum_{y\in \Gamma^k} \sum_{j=0}^{m_y} a_{y j}$ prefix-incomparable codewords and assign them to the $S_y$ so that each $S_y$ contains $a_{y j}$ codewords of length $k-j$ for each $0\le j\le m_y$. This is possible by the Kraft inequality (see \cite[Theorem 5.2.1]{CoverThomas91} or \cite[Theorem 1.6]{Adamek91}), which in this instance is
\begin{align}
\sum_{y\in \Gamma^k}\sum_{j=0}^{m_y} a_{y j}d^{-(k-j)} &\le 1.\label{eq:Kraft}
\end{align}
In fact, equality holds:
\begin{align}
\sum_{y\in \Gamma^k}\sum_{j=0}^{m_y} a_{y j}d^{-(k-j)}
&= \sum_{y\in \Gamma^k}a_{y}d^{-k}
= \sum_{y\in \Gamma^k}p_y = 1. \label{eq:Krafteq}
\end{align}
Each codeword in $S_y$ is of length at most $k$, therefore the capacity is at most $\log d^k = k\log d$.

Since the $d$ symbols of the input process are distributed uniformly, the probability that the input stream begins with a given string of length $n$ is $d^{-n}$.
So $$\Pr(y \pref \ds_k(s,\alpha)) = \Pr(\exists x \in S_y\ x \pref \alpha) =
\sum_{x \in S_y} d^{-|x|} = \sum_{j=0}^{m_y} a_{y j}d^{-(k-j)} = p_y$$ as required, and $D_k$ is indeed a reduction. Moreover, by \eqref{eq:Krafteq}, the probability that a prefix is in some $S_y$ is 1, so the code is exhaustive.

To analyze the efficiency of the simulation, we will use the following lemma.
\begin{lemma}
\label{lem:roundingdigits}
Let the $d$-ary expansion of $a$ be $\sum_{i=0}^m a_id^i$, where $m = \floor{\log_da}$. Then
\(
\left(\log_da - \frac{2d - 1}{d-1}\right)a < \left(m - \frac d{d-1}\right)a < \sum_{i=0}^m ia_id^i \le ma.
\)
\end{lemma}
\begin{proof}
By elementary combinatorics,
\begin{align*}
\sum_{i=0}^{m-1} (m - i)a_id^i &\le \sum_{i=0}^{m-1} (m - i)(d-1)d^i = \frac{d(d^m-1)}{d-1} - m 
< \frac{da}{d-1}.
\end{align*}
Then
\begin{align*}
ma = \sum_{i=0}^m ma_id^i
&\ge \sum_{i=0}^m ia_id^i = ma - \sum_{i=0}^{m-1} (m - i)a_id^i\\
&> \left(m - \frac d{d-1}\right)a = \left(\floor{\log_da} - \frac d{d-1}\right)a\\
&> \left(\log_da - 1 - \frac d{d-1}\right)a = \left(\log_da - \frac{2d - 1}{d-1}\right)a.
\qedhere
\end{align*}
\end{proof}

The expected number of symbols consumed leading to the output $y$ is
\begin{align*}
\sum_{x\in S_y} d^{-\len x}\cdot\len x &= \sum_{j=0}^{m_y} a_{y j}d^{j-k}(k-j)
= kp_y - \sum_{j=0}^{m_y} ja_{y j}d^{j-k}
= kp_y - d^{-k}\sum_{j=0}^{m_y} ja_{y j}d^{j}\\
&< kp_y - d^{-k}\left(\log_d{a_y} - \frac{2d - 1}{d-1}\right)a_y \quad\quad\text{by Lemma \ref{lem:roundingdigits}}\\
&= \frac{2d - 1}{d-1}a_yd^{-k} - p_y\log_d d^{-k} - a_yd^{-k}\log_d{a_y}\\
&= \frac{2d - 1}{d-1}p_y - p_y\log_d{p_y}.
\end{align*}
Thus the expected number of input symbols consumed in one iteration is
\begin{align*}
\sum_{y\in \Gamma^k}\sum_{x\in S_y} d^{-\len x}\cdot\len x 
&<
\sum_{y\in \Gamma^k} \left(\frac{2d - 1}{d-1}p_y - p_y\log_d{p_y}\right)
=
\frac{2d - 1}{d-1} + \frac{H(p_y : y\in\Gamma^k)}{\log d}
\end{align*}
and as the uniform distribution has entropy $\log d$, the expected consumption of entropy is at most
\begin{align}
H(p_y : y\in\Gamma^k) + \log d\cdot\frac{2d - 1}{d-1}
= kH(\seq p1c) + \Theta(1).\label{eq:singlecomponent}
\end{align}
The number of output symbols is $k$, so the production of entropy is $kH(\seq p1c)$. Thus the efficiency is at least
\begin{align*}
\frac{kH(\seq p1c)}{\displaystyle kH(\seq p1c) + \Theta(1)} 
&= \frac{1}{1 + \Theta(k^{-1})} = 1 - \Theta(k^{-1}).
\end{align*}

\begin{example}
Suppose the input distribution is uniform over an alphabet of size $d=24$
and the output alphabet is $u,v,w$ with probabilities $5/24$, $7/24$, and $1/2$ ($=12/24$), respectively.
For $k=2$, there are nine output strings
$uu$, $uv$, $uw$, $vu$, $vv$, $vw$, $wu$, $wv$, $ww$. The string $y$ should be emitted with probability $a_y/24^2$,
where the values of $a_y$ are
25, 35, 60,
35, 49, 84,
60, 84, 144,
respectively. Writing the $a_y$ in base $24$ gives
\begin{align*}
\begin{array}{c@{\hspace{18pt}}c@{\hspace{18pt}}c@{\hspace{18pt}}c@{\hspace{18pt}}c@{\hspace{18pt}}c@{\hspace{18pt}}c@{\hspace{18pt}}c@{\hspace{18pt}}c}
1\hspace{4pt}1 & 1\hspace{4pt}11 & 2\hspace{4pt}12 & 1\hspace{4pt}11 & 2\hspace{4pt}1 & 3\hspace{4pt}12 & 2\hspace{4pt}12 & 3\hspace{4pt}12 & 6\hspace{4pt}0
\end{array}
\end{align*}
respectively. Summing the base-24 digits in 24's place and in 1's place, we obtain 21 and 72, respectively, which means that we need 21 input strings of length one and 72 input strings of length two. As guaranteed by the Kraft inequality \eqref{eq:Kraft}, we can construct an exhaustive prefix code with these parameters, say by taking all 72 strings of length two extending some three strings of length one, along with the remaining 21 strings of length one.

Now we can apportion these to the output strings to achieve the desired probabilities. For example, $uv$ should be emitted with probability $a_{uv}/24^2 = 35/576$, and 35 is $1\;11$ in base 24, which means it should be allocated one input string of length one and 11 input strings of length two. This causes $uv$ to be emitted with the desired probability $1\cdot 24^{-1} + 11\cdot 24^{-2} = 35/576$.

The expected number of input letters consumed in one round is $2\cdot 3/24 + 1\cdot 21/24 = 9/8$ and the entropy of the input distribution is $\log_2 24 \approx 4.59$ for a total entropy consumption of $5.16$ bits. The expected number of output letters produced is $2$ and the entropy of the output distribution is $- \frac 5{24}\log_2 \frac 5{24} - \frac 7{24}\log_2 \frac 7{24} - \frac 12\log_2\frac 12 \approx 1.49$ for a total entropy production of $2.98$ bits. The efficiency is the ratio $2.98/5.16 \approx 0.58$.
\end{example}


\subsection{Uniform $\Imp$ Arbitrary}
\label{sec:dtoany}

Now suppose the target distribution is over an alphabet $\Gamma=\{1,\ldots,c\}$ with arbitrary real probabilities $\seq {p^{(0)}}1c$. It is of course hopeless in general to construct a finite-state protocol with the correct output distribution, as there are only countably many finite-state protocols but uncountably many distributions on $c$ symbols. However, we are able to construct a family of infinite-state restart protocols $D_k$ that map the uniform distribution over a $d$-symbol alphabet $\Sigma$ to the distribution $(p^{(0)}_i : 1\le i\le c)$ with efficiency $1 - \Theta(k^{-1})$. If the probabilities $p^{(0)}_i$ are rational, the resulting protocols $D_k$ will be finite.

Although our formal notion of capacity does not apply to infinite-state protocols, we can still use $k$ as a tunable parameter to characterize efficiency. The restart protocols $D_k$ constructed in this section consist of a serial concatenation of infinitely many component protocols $D_k^{(n)}$, each of capacity $k$. Moreover, $D_k$ is computable if the $p^{(0)}_i$ are, or under the assumption of unit-time real arithmetic; that is, allowing unit-time addition, multiplication, and comparison of arbitrary real numbers.

We assume that $d > c$, which implies that $\max_i p^{(0)}_i > 1/d$. This will ensure that each component has a nonzero probability of emitting at least one output symbol. If $d$ is too small, we can precompose with a protocol from \S\ref{sec:unif} to produce a uniform distribution over a larger alphabet. By Theorem \ref{thm:composition2}, this will not result in a significant loss of efficiency.

The $n$th component $D_k^{(n)}$ of $D_k$ is associated with a real probability distribution $(p^{(n)}_y : y\in\Gamma^k)$ on $k$-symbol output strings. These distributions will be defined inductively. As there are countably many components, in general there will be countably many such distributions, although the sequence will cycle if the original probabilities $p^{(0)}_i$ are rational. The initial component $D_k^{(0)}$ is associated with the target distribution extended to $k$-symbol strings $(p^{(0)}_y : y\in\Gamma^k)$ with $p^{(0)}_{e_1\cdots e_k} = p^{(0)}_{e_1}\cdots p^{(0)}_{e_k}$.

Intuitively, $D_k^{(n)}$ works the same way as in \S\ref{sec:dtorat} using a best-fit rational subprobability distribution $(q^{(n)}_y : y\in\Gamma^k)$ with denominator $d^k$ such that $q^{(n)}_y\le p^{(n)}_y$. If a nonempty string is emitted, the protocol restarts. Otherwise, it passes to $D_k^{(n+1)}$ to handle the residual probabilities. We show that the probability of output in every component is bounded away from 0, so in expectation a finite number of components will be visited before restarting.

Formally, we define $a^{(n)}_y = \floor{p^{(n)}_y d^k}$ and $q^{(n)}_y = a^{(n)}_y d^{-k} \leq p^{(n)}_y$. The rounding error is $p^{(n)}_y-q^{(n)}_y < d^{-k}$. The $q^{(n)}_y$ may no longer sum to 1, and the difference is the \emph{residual probability} $r^{(n)} = 1 - \sum_{y\in \Gamma^k} q^{(n)}_y < (c/d)^k$.

Let $m^{(n)}=r^{(n)}d^k$. As in \S\ref{sec:dtorat}, since 
\begin{align*}
r^{(n)} + \sum_{y\in \Gamma^k} q^{(n)}_y = m^{(n)}d^{-k} + \sum_{y\in \Gamma^k} a^{(n)}_yd^{-k} = 1,
\end{align*}
by the Kraft inequality \eqref{eq:Krafteq}, we can construct an exhaustive prefix code based on the $d$-ary expansions of $m^{(n)}$ and the $a^{(n)}_y$ for $y\in\Gamma^k$ and apportion the codewords to sets $S_y$ and $S_m$ such that the probability of encountering a codeword in $S_y$ is $q^{(n)}_y$ and the probability of encountering a codeword in $S_m$ is $r^{(n)}$. If the protocol encounters a codeword in $S_y$, it emits $y$ and restarts at the start state of $D_k^{(0)}$. If the protocol encounters a codeword in $S_m$, it emits nothing and transitions to the start state of $D_k^{(n+1)}$. The component $D_k^{(n+1)}$ works the same way using the \emph{residual distribution} $(p^{(n+1)}_y : y\in\Gamma^k)$, where
\begin{align*}
p^{(n+1)}_y = \frac{p^{(n)}_y - q^{(n)}_y}{r^{(n)}},
\end{align*}
the (normalized) probability lost when rounding down earlier. An example is given at the end of this section.

To show that the protocol is correct, we need to argue that the string $y$ is emitted with probability $p_y^{(0)}$. It is emitted in $D_k^{(n)}$ with probability $q^{(n)}_y\prod_{j=0}^{n-1}r^{(j)}$, and these are disjoint events, so the probability that $y$ is emitted in any component is $\sum_{n\ge 0} q^{(n)}_y\prod_{j=0}^{n-1}r^{(j)}$.
Using the fact that $q^{(n)}_y=p^{(n)}_y - r^{(n)}p^{(n+1)}_y$,
\begin{align*}
\sum_{n\ge 0} q^{(n)}_y\prod_{j=0}^{n-1}r^{(j)}
&= \sum_{n\ge 0} (p^{(n)}_y - r^{(n)}p^{(n+1)}_y)\prod_{j=0}^{n-1}r^{(j)}\\
&= \sum_{n\ge 0} p^{(n)}_y\prod_{j=0}^{n-1}r^{(j)}
- \sum_{n\ge 0} p^{(n+1)}_y\prod_{j=0}^{n}r^{(j)}
\ =\ p_y^{(0)}.
\end{align*}

We now analyze the production and consumption in one iteration of $D_k$. As just argued, each iteration produces $y\in\Gamma^k$ with probability $p_y^{(0)}$, therefore the entropy produced in one iteration is $H(p_y^{(0)} : y\in\Gamma^k) = kH(\seq{p^{(0)}}1c)$.

To analyze the consumption, choose $k$ large enough that $(c/d)^k \le e^{-1}$ and $(\max_i p_i^{(0)})^k \le e^{-1}$. These assumptions will be used in the following way. If $q\le p \le e^{-1}$, then $-q \log q \le -p\log p$, which can be seen by observing that the derivative
of $-p\log p$
is positive below $e^{-1}$. Thus for any pair of subprobability distributions $(q_n : n\in N)$ and $(p_n : n\in N)$ such that $q_n\le p_n\le e^{-1}$ for all $n\in N$,
\begin{align}
H(q_n : n\in N) \le H(p_n : n \in N).\label{eq:ent4}
\end{align}

Let $s$ be the start state of $D_k$.
Let $V^{(n)}\subs \Sigma^{\omega}$ be the event that the protocol visits the $n$th component $D_k^{(n)}$, and let $U^{(n)}=V^{(n)}\setminus V^{(n+1)}$, the event that the protocol emits a string during the execution of $D_k^{(n)}$. Let $C_k$ be a random variable representing the total consumption in one iteration of $D_k$, and let $C_{kn}$ be a random variable for the portion of $C_k$ that occurs during the execution of $D_k^{(n)}$; that is, $C_{kn}(\alpha)$ is $H(\mu)$ times the number of input symbols read during the execution of $D_k^{(n)}$ on input $\alpha\in\Sigma^\omega$. Note that $C_{kn}(\alpha)=0$ if $\alpha\in U^{(m)}$ and $n>m$, since on input $\alpha$, the single iteration of $D_k$ finishes before stage $n$; thus for $n>m$, the conditional expectation $\Exp{C_{kn} \mid U^{(m)}}=0$.

The expected consumption during a single iteration of $D_k$ is
\begin{align*}
\Exp{C_k}
&= \sum_{m\ge 0} \Pr(U^{{(m)}})\,\Exp{C_k \mid U^{(m)}}
= \sum_{m\ge 0} \Pr(U^{{(m)}})\sum_{n\ge 0}\Exp{C_{kn} \mid U^{(m)}}\\
&= \sum_{n\ge 0}\sum_{m\ge n} \Pr(U^{{(m)}})\,\Exp{C_{kn} \mid U^{(m)}}\\
&= \sum_{n\ge 0} \Pr(\bigcup_{m\ge n} U^{{(m)}})\,\Exp{C_{kn} \mid \bigcup_{m\ge n}U^{{(m)}}}\\
&= \sum_{n\ge 0} \Pr(V^{{(n)}})\,\Exp{C_{kn} \mid V^{{(n)}}}.
\end{align*}
Since $r^{(n)} \le (c/d)^k$,
\begin{align*}
\Pr(V^{(n)}) = \prod_{j=0}^{n-1} r^{(j)} \le (\frac cd)^{kn}.
\end{align*}
The quantity $\Exp{C_{kn} \mid V^{{(n)}}}$ is the expected consumption during the execution of $D_k^{(n)}$, conditioned on the event that $D_k^{(n)}$ is visited. By \eqref{eq:singlecomponent}, this is at most
\begin{align*}
H(q_y^{(n)} : y\in\Gamma^k) - r^{(n)}\log r^{(n)} + \Theta(1)
&= H(q_y^{(n)} : y\in\Gamma^k) + \Theta(1),
\end{align*} 
since $r^{(n)}\le (c/d)^k \le e^{-1}$, therefore $- r^{(n)}\log r^{(n)} \le -e^{-1}\log e^{-1}=\Theta(1)$.
Using the naive upper bound $H(q_y^{(n)} : y\in\Gamma^k)\le k\log c$ from the uniform distribution for $n\ge 1$, we have
\begin{align*}
\cdot\Exp{C_k}
&= \sum_{n\ge 0} \Pr(V^{{(n)}})\,\Exp{C_{kn} \mid V^{{(n)}}}\\
&\le \sum_{n\ge 0}(\frac cd)^{kn}(H(q_y^{(n)} : y\in\Gamma^k) + \Theta(1))\\
&= \sum_{n\ge 1}(\frac cd)^{kn}H(q_y^{(n)} : y\in\Gamma^k) + H(q_y^{(0)} : y\in\Gamma^k) + \Theta(1)\\
&\le \frac{(c/d)^k}{(1-(c/d)^k)}k\log c + H(q_y^{(0)} : y\in\Gamma^k) + \Theta(1)\\
&\le kH(\seq{p^{(0)}}1c) + \Theta(1).
\end{align*}
The last inference holds because
\begin{align*}
H(q_y^{(0)} : y\in\Gamma^k) \le H(p_y^{(0)} : y\in\Gamma^k) = kH(\seq{p^{(0)}}1c)
\end{align*}
as justified by \eqref{eq:ent4}.

The efficiency is the ratio of production to consumption, which is at least
\begin{align*}
\frac{kH(\seq{p^{(0)}}1c)}{kH(\seq{p^{(0)}}1c) + \Theta(1)}
&= \frac{1}{1 + \Theta(k^{-1})}
= 1 - \Theta(k^{-1}).
\end{align*}

There is still one issue to resolve if we wish to construct a serial protocol with $k$th component $D_k$. As $D_k$ is not finite-state, its consumption is not uniformly bounded by any $m_k$, as required by Lemma \ref{lem:variance}. However, one iteration of $D_k$ visits a series of components, and the consumption in each component is uniformly bounded. In each component, if output is produced, the protocol restarts from the first component, otherwise the computation proceeds to the next component. Each component, when started in its start state, consumes at most $k$ digits and produces exactly $k$ digits with probability at least $1-(c/d)^k$ and produces no digits with probability at most $(c/d)^k$. The next lemma shows that this is enough to derive the conclusion of Lemma \ref{lem:variance}.

\begin{lemma}
Let $m_k$ be a uniform bound on the consumption in each component $D_k^{(n)}$ of one iteration of $D_k$. If the $m_k$ satisfy \eqref{eq:cassm}, then the variances \eqref{eq:variance} and \eqref{eq:variance1} vanish in the limit.
\end{lemma}
\begin{proof}
As above, let $C_k$ be a random variable for the consumption in one iteration of $D_k$, let $C_{kn}$ be a random variable for the consumption in $D_k^{(n)}$, and let $U^{(n)}$ be the event that $D_k^{(n)}$ produces a nonnull string. Again using the fact that $C_{kn}$ restricted to $U^{(m)}$ is $0$ for $n>m$,
\begin{align*}
\Exp{C_k^2 \mid U^{(m)}}
&= \Exp{(\sum_{n\ge 0} C_{kn})^2 \mid U^{(m)}}\\
&= \Exp{\sum_{n\ge 0} C_{kn}^2 \mid U^{(m)}} + 2\Exp{\sum_{0\le n < \ell\le m} C_{kn}C_{k\ell} \mid U^{(m)}}\\
&\le m_k\Exp{C_{k} \mid U^{(m)}} + 2m_k^2\binom{m+1}2.
\end{align*}
Then $\Pr(U^{(m)})\le\Pr(V^{(m)})\le(c/d)^{km}$ and
\begin{align*}
\Exp{C_k^2} &= \sum_{m\ge 0} \Pr(U^{(m)})\cdot\Exp{C_k^2 \mid U^{(m)}}\\
&\le m_k\sum_{m\ge 0}\Pr(U^{(m)})\Exp{C_{k} \mid U^{(m)}} + 2m_k^2\sum_{m\ge 0} \Pr(U^{(m)})\binom{m+1}2\\
&\le m_k\Exp{C_k} + 2m_k^2\sum_{m\ge 0} (\frac cd)^{km}\frac{(m+1)m}2\\
&= m_kc_k + 2m_k^2(c/d)^k(1-(c/d)^k)^{-3}
= m_kc_k + o(1),
\end{align*}
\begin{align*}
\Var{\frac{\sum_{k=0}^n C_k}{\sum_{k=0}^n c_k}}
&= \frac{\sum_{k=0}^n \Var{C_k}}{(\sum_{k=0}^n c_k)^2}
\le \frac{\sum_{k=0}^n \Exp{C_k^2}}{(\sum_{k=0}^n c_k)^2}\\
&\le \frac{m_n}{\sum_{k=0}^n c_k}\cdot\frac{\sum_{k=0}^n c_k}{\sum_{k=0}^n c_k} + o(1)
= o(1).
\qedhere
\end{align*}
\end{proof}

\begin{example}
Consider the case $d=6$, $c=2$, and $k=1$ in which the output letters $u,v$ should be emitted with probability $p$ and $1-p$, respectively. The input distribution is a fair six-sided die. We will try to find a best-fit rational distribution with denominator $6$.

In the first component, we roll the die with result $n\in\{1,\ldots,6\}$ and emit $u$ if $n/6\le p$ and $v$ if $(n-1)/6\ge p$. Thus $u$ is emitted with probability $\floor{6p}/6\le p$ and $v$ with probability $1-\ceil{6p}/6\le 1-p$. If $p\in\{1/6,2/6,\ldots,6/6\}$, then exactly one of those two events occurs. In this case there are no further components, as $u$ and $v$ have been emitted with the desired probabilities $p$ and $1-p$, respectively; the residual probabilities are 0. The protocol restarts in the start state of the first component.

Otherwise, if $p\not\in\{1/6,2/6,\ldots,6/6\}$, then $u$ and $v$ are emitted with probability $\floor{6p}/6<p$ and $1-\ceil{6p}/6<1-p$ respectively, and nothing is emitted with probability $1/6$, which happens when $n=\ceil{6p}$. In the event nothing is emitted, we move on to the second component, which is exactly like the first except with the residual probabilities $p'=6(p-\floor{6p}/6)=6p-\floor{6p}$ and $1-p'=6(\ceil{6p}/6-p)=6\ceil{6p}-6p$. The factor $6$ appears because we are conditioning on the event that no symbol was emitted in the first component, which occurs with probability 1/6. We continue in this fashion as long as there is nonzero residual probability.

For a concrete instance, suppose $p=16/215$ and $1-p=199/215$. Since $p$ falls in the open interval $(0,1/6)$, we will emit $v$ if the die roll is 2, 3, 4, 5, or 6 and emit nothing if the die roll is 1. In the latter event, we move on to the second component using the residual probabilities $p'=6\cdot 16/215-\floor{6\cdot 16/215}=96/215$ and $1-p'=119/215$. Since $p'\in(1/3,1/2)$, we will emit $u$ if the die roll is 1 or 2, $v$ if it is 4, 5, or 6, and nothing if it is 3. In the last event, we move on to the third component using the residual probabilities $p''=6\cdot 96/215 - \floor{6\cdot 96/215}=146/215$ and $1-p'=69/215$. Since $p''\in(2/3,5/6)$, we will emit $u$ if the die roll is 1, 2, 3, or 4, $v$ if it is 6, and nothing if it is 5. In the last event, we move on to the fourth component using the residual probabilities $p'''=6\cdot 146/215 - \floor{6\cdot 146/215}=16/215$ and $1-p'''=199/215$.

Note that after three components, we have $p'''=p$, so the fourth component is the same as the first. We are are back to the beginning and can return to the first component. In general, the process will eventually cycle iff the probabilities are rational. This gives an alternative to the construction of \S\ref{sec:dtorat}. Note also that at this point, $u$ has been emitted with probability $(1/6)(2/6) + (1/6^2)(4/6)=2/17=16/216$ and $v$ with probability $5/6 + (1/6)(3/6) + (1/6^2)(1/6) = 199/216$. These numbers are proportional to $p$ and $1-p$, respectively, out of $215/216$, the probability that some symbol has been emitted.
\end{example}


\subsection{Arbitrary $\Rightarrow$ $c$-Uniform with Efficiency $1-\Theta(\log k/k)$}
\label{sec:arb2unif}

In this section, we describe a family of restart protocols $B_k$ for transforming an arbitrary $d$-ary distribution with real probabilities $\seq p1d$ to a $c$-ary uniform distribution with $\Theta(\log k/k)$ loss. Unlike the other protocols we have seen so far, these protocols do not depend on knowledge of the input distribution; perhaps as a consequence of this, the convergence is asymptotically slower by a logarithmic factor.

Let $D=\{1,\ldots,d\}$ be the input alphabet. Let $G_k$ be the set of all sequences $\sigma\in\naturals^D$ such that $\sum_{i\in D}\sigma_i=k$. Each string $y\in D^k$ is described by some $\sigma\in G_k$, where $\sigma_i$ is the number of occurrences of $i\in D$ in $y$. Let $V_\sigma$ be the set of strings in $D^k$ whose letter counts are described by $\sigma$ in this way.

The protocol $B_k$ works as follows. Make $k$ calls on the input distribution to obtain a $d$-ary string of length $k$. The probability that the string is in $V_\sigma$ is $q_\sigma = \len{V_\sigma}p_\sigma$, where
\begin{align*}
\len{V_\sigma} &= \binom k{\sigma_1\ \ldots\ \sigma_d} & p_\sigma &= \prod_{i\in D} p_i^{\sigma_i},
\end{align*}
as there are $\len{V_\sigma}$ strings in $D^k$ whose letter counts are described by $\sigma$, each occurring with probability $p_\sigma$. Thus the
strings in $V_\sigma$ are distributed uniformly. For each $\sigma$, apply the protocol $P_k$ of \S\ref{sec:unif} to the elements of $V_\sigma$ to produce $c$-ary digits, thereby converting the $\len{V_\sigma}$-uniform distribution on $V_\sigma$ to a $c$-uniform distribution.
The states then just form the $d$-ary tree of depth $k$ that stores the input, so the capacity of $B_k$ is approximately $k\log d$. 

To analyze the efficiency of $B_k$, we can reuse an argument from \S\ref{sec:unif}, with the caveat that the size of the input alphabet was a constant there, whereas $\len{V_\sigma}$ is unbounded. Nevertheless, we were careful in \S\ref{sec:unif} that the part of the argument that we need here did not depend on that assumption.

For each $\sigma$, let $m=\floor{\log_c \len{V_\sigma}}$ and let the $c$-ary expansion of $\len{V_\sigma}$ be
\begin{align*}
\len{V_\sigma} &= \sum_{i=0}^{m} a_i c^i,
\end{align*}
where $0\le a_i \le c-1$ and $a_m\neq 0$. It was established in \S\ref{sec:unif}, equation \eqref{eq:a4}, that the expected number of $c$-ary digits produced by strings in $V_\sigma$ is at least
\begin{align*}
\floor{\log_c \len{V_\sigma}} - \frac{c}{c-1}
&\ge (\log_c\len{V_\sigma}-1) - \frac{c}{c-1}
= \log_c\len{V_\sigma} - \frac{2c-1}{c-1},
\end{align*}
thus the expected number of $c$-ary digits produced in all is at least
\begin{align*}
\sum_\sigma q_\sigma \log_c \len{V_\sigma} - \frac{2c-1}{c-1}.
\end{align*}
The total entropy production is this quantity times $\log c$, or
\begin{align*}
& \sum_\sigma q_\sigma \log_2\len{V_\sigma} - b,
\end{align*}
where $b = (2c-1)\log c/(c-1)$.

The total entropy consumption is $kH(\seq p1d)$. This can be viewed as the composition of a random choice that chooses the number of occurrences of each input symbol followed by a random choice that chooses the arrangement of the symbols. Using the conditional entropy rule (Lemma \ref{lem:condent}),
\begin{align}
kH(\seq p1d) &= H(q_\sigma \mid \sigma\in G_k) + \sum_\sigma q_\sigma\log \len{V_\sigma}\nonumber\\
&\le \log\binom{k+d-1}{k} + \sum_\sigma q_\sigma\log \len{V_\sigma}\label{eq:arbunif1}\\
&\le d\log k + \sum_\sigma q_\sigma\log \len{V_\sigma},\label{eq:arbunif2}
\end{align}
for sufficiently large $k$.
The inequality \eqref{eq:arbunif1} comes from the fact that the entropy of the uniform distribution exceeds the entropy of any other distribution on the same number of letters.
The inequality \eqref{eq:arbunif2} comes from the fact that the binomial expression is bounded by $(k+1)^{d-1}$, which is bounded by $k^d$ for all $k$ such that $k\ln k\ge d-1$.

Dividing, the production/consumption ratio is
\begin{align*}
\frac{\sum_\sigma q_\sigma \log \len{V_\sigma} - b}{kH(\seq p1d)}
&= \frac{d\log k + \sum_\sigma q_\sigma \log \len{V_\sigma}}{kH(\seq p1d)} - \frac{d\log k + b}{kH(\seq p1d)}\\
&\ge \frac{d\log k + \sum_\sigma q_\sigma \log \len{V_\sigma}}{d\log k + \sum_\sigma q_\sigma \log \len{V_\sigma}} - \frac{d\log k + b}{kH(\seq p1d)}\\
&= 1 - \Theta(\log k/k).
\end{align*}

\begin{example}
Consider the case $d=3$, $c=2$, and $k=5$ in which the three-letter input alphabet is $u,v,w$ with probabilities $p,q,r$ respectively, and the output distribution is a fair coin. We partition the input strings of length five into disjoint classes $V_\sigma$ depending on the number of occurrences of each input symbol. There are $\binom{5+2}{2}=21$ classes represented by the patterns in the following table:
\begin{align*}
\begin{array}{|c|c|c|c|c|c|}\hline
& \textsl{pattern} & \textsl{no.\ of instances} & \textsl{no.\ of classes} & \textsl{production} & \textsl{probability}\\\hline\hline
1 & 4,1 & 5 & 6 & 8/5 & 30\cdot 3^{-5}\\ \hline
2 & 3,2 & 10 & 6 & 13/5 & 60\cdot 3^{-5}\\ \hline
3 & 3,1,1 & 20 & 3 & 18/5 & 60\cdot 3^{-5}\\\hline
4 & 2,2,1 & 30 & 3 & 49/15 & 90\cdot 3^{-5}\\\hline
5 & 5 & 1 & 3 & 0 & 3\cdot 3^{-5}\\\hline
\end{array}
\end{align*}
Row 1 represents classes consisting of all strings with four occurrences of one letter and one occurrence of another. Each such class has five instances, depending on the arrangement of the letters. For example, the five instances of strings containing four occurrences of $u$ and one of $v$ are $uuuuv$, $uuuvu$, $uuvuu$, $uvuuu$, and $vuuuu$. Each of these five instances occurs with the same probability $p^4q$, so this class can be used as a uniformly distributed source over a five-letter alphabet. There are six classes of this form, corresponding to the six choices of two letters.

Similarly, row 2 represents classes with three occurrences of one letter and two of another. Each such class has $\binom{5}{3\ 2}=10$ instances, all of which occur with the same probability, so each such class can be used as a uniformly distributed source over a ten-letter alphabet. There are six such classes. The classes of rows 3 and 4 can be used as uniformly distributed sources over 20- and 30-letter alphabets, respectively. The classes in row 5 have only one instance and are not usable.

In each round of the protocol, we sample the input distribution five times. Depending on the class of the resulting string, we apply one of the protocols of \S\ref{sec:unif} to convert to fair coin flips. For example, if the string is in one of the classes from row 3 above, which is a uniform source on $\binom{5}{3\,1\,1}=20$ letters, writing 20 in binary gives 10100, indicating that 16 of the 20 instances should emit the 16 binary strings of length four, and the remaining four instances should emit 00, 01, 10, and 11, respectively. The expected production is $4\cdot 16/20 + 2\cdot 4/20 = 18/5$. This will be the production for any class in row 3, which will transpire with probability
\begin{align*}
20(p^3qr + pq^3r + pqr^3),
\end{align*}
the probability that the input string falls in a class in row 3.

The last column of the table lists these values for the case $p=q=r=1/3$.
In that case, the total consumption is $5\log 3\approx 7.92$. The production is
\begin{align*}
& 30\cdot 3^{-5}\cdot 8/5
+ 60\cdot 3^{-5}\cdot 13 / 5
+ 60\cdot 3^{-5}\cdot 18/5
+ 90\cdot 3^{-5}\cdot 49/15\ \approx\ 2.94,
\end{align*}
for an efficiency of $2.94/7.92 \approx 0.37$.
\end{example}

There is much wasted entropy with this scheme for small values of $k$. The sampling can be viewed as a composition of a first stage that selects the class, followed by a stage that selects the string within the class. All of the entropy consumed in the first stage is wasted, as it does not contribute to production.


\subsection{$(\frac 1r,\frac{r-1}r)$ $\Rightarrow$ $(r-1)$-Uniform with Efficiency $1-\Theta(k^{-1})$}
\label{sec:coin2unif}

Let $r\in\naturals$, $r\ge 3$. In this section we show that a coin with bias $1/r$ can generate an $(r-1)$-ary uniform distribution with $\Theta(k^{-1})$ loss of efficiency. This improves the result of the previous section in this special case.

Dirichlet's approximation theorem (see \cite{Cassels57,Lang95,Schmidt96}) states that for irrational $u$, there exist infinitely many pairs of integers $k,m > 0$ such that $\len{ku - m} < 1/k$. We need the result in a slightly stronger form.
\begin{lemma}
\label{lem:dense}
Let $u$ be irrational. For infinitely many integers $k > 0$,
$ku - \floor{ku} < \frac 1{k+1}$.
\end{lemma}
\begin{proof}
The numbers $ku - \floor{ku}$, $k \ge 1$, are all distinct since $u$ is irrational. In the following, we use real arithmetic modulo 1, thus we write $iu$ for $iu-\floor{iu}$ and $0u$ for both $0$ and $1$.

Imagine placing the elements $u,2u,3u,\ldots$ in the unit interval one at a time, $iu$ at time $i$. At time $k$, we have placed $k$ elements, which along with 0 and 1 partition the unit interval into $k+1$ disjoint subintervals. We make three observations:
\begin{enumerate}
\romanize
 \item At time $k$, the smallest interval is of length less than $\frac 1{k+1}$.
 \item An interval of minimum length always occurs adjacent to $0$ or $1$.
 \item Let $k_0,k_1,k_2,\ldots$ be the times at which the minimum interval length strictly decreases. For all $i$, the new smallest interval created at time $k_i$ is adjacent to 0 iff the new smallest interval created at time $k_{i+1}$ is adjacent to 1.
\end{enumerate}
For (i), the average interval length is $\frac 1{k+1}$, so there must be one of length less than that. It cannot be exactly $\frac 1{k+1}$ because $u$ is irrational.

For (ii), suppose $[iu,ju]$ is a minimum-length interval.
If $i < j$, then the interval $[0,(j-i)u]$ is the same length and was created earlier.
If $i > j$, then the interval $[(i-j)u,1]$ is the same length and was created earlier. Thus the first time a new minimum-length interval is created, it is created adjacent to either 0 or 1.

For (iii), we proceed by induction. The claim is certainly true after one step. Now consider the first time a new minimum-length interval is created, say at time $k$. Let $[iu,1]$ be the interval adjacent to 1 and $[0,ju]$ the interval adjacent to 0 just before time $k$. Suppose that $[iu,1]$ is the smaller of the two intervals (the other case is symmetric). By the induction hypothesis, $j < i$. By (ii), either $iu < ku < 1$ or $0 < ku < ju$. But if the former, then $ku - iu = (k-i)u$ and $0 < (k-i)u < ju$, a contradiction, since then $[0,(k-i)u]$ would be a smaller interval adjacent to $0$.

By (i)--(iii), every other time $k$ that a new minimum-length interval is created, it is adjacent to $0$ and its length is less than $\frac 1{k+1}$.
\end{proof}

Choose $k\ge r-2$ and $m=\floor{k\log_{r-1} r}$. Note that $\log_{r-1}r$ is irrational: if $\log_{r-1}r = p/q$ then $r^q = (r-1)^p$, which is impossible because $r$ and $r-1$ are relatively prime.
Then $(r-1)^m < r^k$ and $m > k$ for sufficiently large $k$.
We have two representations of $\frac{r^k - 1}{r - 1}$ as a sum:
\begin{align*}
\frac{r^k - 1}{r-1} &= \sum_{i=0}^{k-1} r^i = \sum_{i=0}^{k-1} \binom k{i+1} (r-1)^i.
\end{align*}
Moreover, every integer in the interval $[0,\frac{r^k - 1}{r-1}]$ can be represented by a sum of the form
$\sum_{i=0}^{k-1} a_i (r-1)^i$,
where $0\le a_i\le \binom k{i+1}$. (We might call this a \emph{binomialary representation}.) To see this, let
$t$ be any number less than $\frac{r^k - 1}{r-1}$ with such a representation, say
$ t = \sum_{i=0}^{k-1} a_i (r-1)^i $.
We show that $t+1$ also has such a representation. Let $i$ be the smallest index such that $a_i<\binom k{i+1}$. Then
\begin{align*}
t &= \left(\sum_{j=0}^{i-1} \binom k{j+1} (r-1)^j\right) + \left(\sum_{j=i}^{k-1} a_j (r-1)^j\right) &
1 &= (r-1)^i - \sum_{j=0}^{i-1} (r-2) (r-1)^j.
\end{align*}
Adding these, we have
\begin{align*}
t+1 &= \left(\sum_{j=0}^{i-1} (\binom k{j+1}-r+2) (r-1)^j\right) + (a_i+1)(r-1)^i + \left(\sum_{j=i+1}^{k-1} a_j (r-1)^j\right),
\end{align*}
and this is of the desired form.

It follows that every multiple of $r-1$ in the interval $[0,r^k - 1]$ can be represented by a sum of the form
$\sum_{i=0}^{k} a_i (r-1)^i$ with $0\le a_i\le \binom ki$. In particular, $(r-1)^m$ can be so represented. Thus
\begin{align}
(r-1)^m &= \sum_{i=0}^{k} a_i (r-1)^i,\label{eq:r1mrep}
\end{align}
where $0\le a_i\le \binom ki$. 

Pick $k>\ln(r-1)-1$, which ensures that $0<\ln(r-1)/(k+1)<1$, and also large enough that
\begin{align}
k\log_{r-1}r - \floor{k\log_{r-1}r} &< \frac 1{k+1},\label{eq:r1mrep2}
\end{align}
which is possible by Lemma \ref{lem:dense}. Using the fact that $\ln x \le x-1$ for all $x > 0$,
\begin{align*}
k\log_{r-1}r - m &= k\log_{r-1}r - \floor{k\log_{r-1}r}
< \frac 1{k+1} = \frac{\log_{r-1}e\cdot\ln(r-1)}{k+1}\\
&\le -\log_{r-1}e\cdot\ln(1 - \frac{\ln(r-1)}{k+1})
= -\log_{r-1}(1 - \frac{\ln(r-1)}{k+1}).
\end{align*}
Rearranging terms and exponentiating, we obtain
\begin{align}
\frac{(r-1)^m}{r^{k}} &\ge 1 - \frac{\ln(r-1)}{k+1} = 1 - \Theta(k^{-1}).\label{eq:r1mrep3}
\end{align}
From \eqref{eq:r1mrep}, we have $1 = \sum_{i=0}^{k} a_i (r-1)^{i-m}$.
Thus we can find an exhaustive prefix code $A$ over the $(r-1)$-ary target alphabet with exactly $a_i$ words of length $m-i$, $0\le i\le k$. Assign a distinct word over the binary source alphabet of length $k$ and probability $(r-1)^i/r^k$ to each codeword of length $m-i$ so that the mapping from input words to codewords is injective. There are enough input words to do this, as we need $a_i$ input words of probability $(r-1)^i/r^k$, and there are $\binom ki\ge a_i$ such input words in all.

There are $(r-1)^m$ words over the source alphabet with a target word assigned to them. If one of these source words comes up in the protocol, output its associated target word. For the remaining $r^k - (r-1)^m$ source words, do not output anything. This is lost entropy.

To argue that the output distribution is uniform, we first show that for every prefix $x$ of a codeword in $A$, $x$ appears as a prefix of an emitted codeword with probability $(r-1)^{m-\len x}/r^k$.

We proceed by reverse induction on $\len x$. The claim is true for codewords $x\in A$ by construction. Since $A$ is exhaustive, for every proper prefix $x$ of a codeword and every letter $c$, $xc$ is also a prefix of a codeword. Each such $xc$ is emitted as a prefix with probability $(r-1)^{m-\len{xc}}/r^k$ by the induction hypothesis, and these events are disjoint, therefore $x$ is emitted as a prefix with probability
\begin{align*}
\sum_c \frac{(r-1)^{m-\len{xc}}}{r^k}
\ =\ (r-1)\frac{(r-1)^{m-\len{x}-1}}{r^k}
\ =\ \frac{(r-1)^{m-\len{x}}}{r^k}.
\end{align*}

It follows that every letter $c$ appears as the $n^{\mathrm th}$ letter of an emitted codeword with the same probability $\len{A_{n-1}}\cdot(r-1)^{m-n}/r^k$, where $A_{n-1}$ is the set of length-$(n-1)$ proper prefixes of target codewords, therefore the distribution is uniform.

To calculate the efficiency, by elementary combinatorics, we have
\begin{align*}
\sum_{i=0}^{k} ia_i (r-1)^ir^{-k} &\le \sum_{i=0}^{k} i\binom ki (r-1)^ir^{-k} = k\frac{r-1}r.
\end{align*}
Using \eqref{eq:r1mrep3}, the expected number of target symbols produced is
\begin{align*}
\sum_{i=0}^{k} (m-i)a_i (r-1)^ir^{-k} 
&= \sum_{i=0}^{k} ma_i (r-1)^ir^{-k} - \sum_{i=0}^{k} ia_i (r-1)^ir^{-k}\\
&\ge m(r-1)^mr^{-k} - k\frac{r-1}r\\
&\ge m(1-\Theta(k^{-1})) - k\frac{r-1}r
= m - k\frac{r-1}r - \Theta(1),
\end{align*}
as $m$ is $\Theta(k)$. The number of source symbols consumed is $k$.

The information-theoretic bound on the production/consumption ratio is the quotient of the source and target entropies:
\begin{align*}
\frac{\frac 1r\log r + \frac{r-1}r\log\frac{r}{r-1}}{\log(r-1)} &= \log_{r-1} r - \frac{r-1}r.
\end{align*}
We also have
\begin{align*}
\frac mk &= \frac{\floor{k\log_{r-1}r}}k > \frac{{k\log_{r-1}r} - 1}k > \log_{r-1}r - \Theta(k^{-1}).
\end{align*}
The production/consumption ratio is thus
\begin{align*}
{\frac mk - \frac{r-1}r - \Theta(k^{-1})} 
&> \log_{r-1}r - \frac{r-1}r - \Theta(k^{-1}),
\end{align*}
which is within $\Theta(k^{-1})$ of optimal.

\begin{example}
Consider the case $r=k=4$ in which the input alphabet is $u,v$ with probabilities $1/4$ and $3/4$ respectively, and the output distribution is uniform on the ternary alphabet $0,1,2$. Then $\log_{r-1}r=\log_3 4\approx 1.26$ and $m = \floor{4\log_3 4} = 5$. The conditions for applying the protocol are satisfied: $m > k$, $4 > \ln 3 - 1$, and as required by \eqref{eq:r1mrep2}, $4\log_34 - \floor{4\log_34} \approx .05 < 1/4$.

As guaranteed by \eqref{eq:r1mrep}, we can write $3^5 = 243$ as $\sum_{i=0}^4 a_i3^i$ with $0\le a_i\le \binom 4i$. The coefficients $a_0=a_1=0$, $a_2=6$, $a_3=4$, and $a_4=1$ do the trick. (The representation is not unique; the coefficients 0, 3, 5, 4, 1 will work as well.)
We now select an exhaustive prefix code over the ternary output alphabet with exactly $a_i$ codewords of length $5-i$. The code
\begin{align*}
& 0 && 10 && 11 && 12 && 20 && 210 && 211 && 212 && 220 && 221 && 222
\end{align*}
does it. Now we assign to each codeword of length $5-i$ a distinct input word of length $4$ and probability $3^i/4^k = (3/4)^i(1/4)^{k-i}$. We can assign them 
\begin{align*}
& v^4 && uv^3 && vuv^2 && v^2uv && v^3u && u^2v^2 && uvuv && uv^2u && vu^2v && vuvu && v^2u^2
\end{align*}
respectively.
The following diagram shows the output words and the probabilities with which they are emitted:
\begin{align*}
\begin{tikzpicture}[-, node distance=7mm, auto]
\tikzstyle{node}=[circle, inner sep=1.6pt, fill=black]
\small
  \node[node] (e) {};
  \node[node, below of=e, xshift=-2mm] (1) {};
  \node[node, left of=1, node distance=15mm] (0) {};
  \node[node, right of=1, node distance=21mm] (2) {};
  \node[node, below of=1, xshift=-1mm] (11) {};
  \node[node, left of=11, node distance=8mm] (10) {};
  \node[node, right of=11, node distance=8mm] (12) {};
  \node[node, below of=2, xshift=6mm] (21) {};
  \node[node, left of=21, node distance=8mm] (20) {};
  \node[node, right of=21, node distance=14mm] (22) {};
  \node[node, below of=21, xshift=0mm] (211) {};
  \node[node, left of=211, node distance=5mm] (210) {};
  \node[node, right of=211, node distance=5mm] (212) {};
  \node[node, below of=22, xshift=1mm] (221) {};
  \node[node, left of=221, node distance=5mm] (220) {};
  \node[node, right of=221, node distance=5mm] (222) {};
  \node[below of=0, node distance=12pt] (a0) {{$\frac{3^4}{4^4}$}};
  \node[below of=10, node distance=12pt] (a10) {{$\frac{3^3}{4^4}$}};
  \node[below of=11, node distance=12pt] (a11) {{$\frac{3^3}{4^4}$}};
  \node[below of=12, node distance=12pt] (a12) {{$\frac{3^3}{4^4}$}};
  \node[below of=20, node distance=12pt,xshift=-2pt] (a20) {{$\frac{3^3}{4^4}$}};
  \node[below of=210, node distance=12pt] (a210) {{$\frac{3^2}{4^4}$}};
  \node[below of=211, node distance=12pt] (a211) {{$\frac{3^2}{4^4}$}};
  \node[below of=212, node distance=12pt] (a212) {{$\frac{3^2}{4^4}$}};
  \node[below of=220, node distance=12pt] (a220) {{$\frac{3^2}{4^4}$}};
  \node[below of=221, node distance=12pt] (a221) {{$\frac{3^2}{4^4}$}};
  \node[below of=222, node distance=12pt] (a222) {{$\frac{3^2}{4^4}$}};
  \path (e) edge node[swap, xshift=-7pt, yshift=-6pt] {\scriptsize{$0$}} (0);
  \path (e) edge node[xshift=-1pt, yshift=6pt] {\scriptsize{$1$}} (1);
  \path (e) edge node[xshift=10pt, yshift=-6pt] {\scriptsize{$2$}} (2);
  \path (1) edge node[swap, xshift=-3pt, yshift=-6pt] {\scriptsize{$0$}} (10);
  \path (1) edge node[swap, xshift=1pt, yshift=-6pt] {\scriptsize{$1$}} (11);
  \path (1) edge node[xshift=3pt, yshift=-6pt] {\scriptsize{$2$}} (12);
  \path (2) edge node[swap, xshift=1pt, yshift=-6pt] {\scriptsize{$0$}} (20);
  \path (2) edge node[xshift=1pt, yshift=-6pt] {\scriptsize{$1$}} (21);
  \path (2) edge node[xshift=10pt, yshift=-6pt] {\scriptsize{$2$}} (22);
  \path (21) edge node[swap, xshift=-1pt, yshift=-6pt] {\scriptsize{$0$}} (210);
  \path (21) edge node[swap, xshift=2pt, yshift=0pt] {\scriptsize{$1$}} (211);
  \path (21) edge node[xshift=1pt, yshift=-6pt] {\scriptsize{$2$}} (212);
  \path (22) edge node[swap, xshift=-1pt, yshift=-6pt] {\scriptsize{$0$}} (220);
  \path (22) edge node[xshift=-2pt, yshift=-6pt] {\scriptsize{$1$}} (221);
  \path (22) edge node[xshift=2pt, yshift=-6pt] {\scriptsize{$2$}} (222);
\end{tikzpicture}
\end{align*}
Note that $0$, $1$, and $2$ are each emitted as the $n^{\mathrm{th}}$ letter with equal probability, so the distribution is uniform. For $n=1$, the probabilities are
\begin{align*}
& \textstyle \frac{3^4}{4^4}\ =\ 3\cdot\frac{3^3}{4^4}\ =\ 1\cdot\frac{3^3}{4^4} + 6\cdot\frac{3^2}{4^4}\ .
\end{align*}
Likewise, for $n=2$ and $n=3$, the probabilities are
\begin{align*}
& \textstyle 2\cdot\frac{3^4}{4^4}\ =\ \frac{3^3}{4^4} + 3 \cdot \frac{3^2}{4^4}
&& \textstyle 2\cdot \frac{3^2}{4^4}\ ,
\end{align*}
respectively.

The entropy consumption is
\begin{align*}
& \textstyle 4(-\frac 14\log\frac 14 - \frac 34\log \frac 34)\ =\ 8 - 3\log 3\ \approx\ 3.25
\end{align*}
and the production is
\begin{align*}
& \textstyle (1\cdot1\cdot\frac{3^4}{4^4} + 4\cdot 2\cdot\frac{3^3}{4^4} + 6\cdot 3\cdot\frac{3^2}{4^4})\log 3
\ =\ \frac{459}{256}\log 3
\ \approx\ 2.84
\end{align*}
for an efficiency of $2.84/3.25 \approx 0.87$. The alternative coefficients 0, 3, 5, 4, 1 give slightly better production of $\frac{486}{256}\log 3 \approx 3.01$ for the same consumption, yielding an improved efficiency of $3.01/3.25 \approx 0.93$.
\end{example}

\section{Conclusion}

We have introduced a coalgebraic model for constructing and reasoning about state-based protocols that implement entropy-conserving reductions between random processes. We have provided basic tools that allow efficient protocols to be constructed in a compositional way and analyzed in terms of the tradeoff between state and loss of entropy. We have illustrated the use of the model in various reductions.

An intriguing open problem is to improve the loss of the protocol of \S\ref{sec:arb2unif} to $\Theta(1/k)$. Partial progress has been made in \S\ref{sec:coin2unif}, but we were not able to generalize this approach.

\subsection{Discussion: The Case for Coalgebra}

What are the benefits of a coalgebraic view? Many constructions in the information theory literature are expressed in terms of trees; e.g.~\cite{Hirschler17,Boecherer14}. Here we have defined protocols as coalgebras $(S,\dd)$, where $\dd:S\times\Sigma\to S\times\Gs$, a form of Mealy automata. These are not trees in general. However, the class admits a final coalgebra
$D:(\Gs)^{\Sp}\times\Sigma\to(\Gs)^{\Sp}\times\Gs$,
where
\begin{align*}
D(f,a) &= (f\rest a,f(a)) & f\rest a(x) &= f(ax),\ a\in\Sigma,\ x\in\Sp.
\end{align*}
Here the extension to streams $D^\omega:(\Gs)^{\Sp}\times\So\pfun\Go$ takes the simpler form
\begin{align*}
D^\omega(f,a\alpha)
= {f(a)\cdot D^\omega(f\rest a,\alpha)}.
\end{align*}
A state $f:\Sp\to\Gs$ can be viewed as a labeled tree with nodes $\Ss$ and edge labels $\Gs$. The nodes $xa$ are the children of $x$ for $x\in\Ss$ and $a\in\Sigma$. The label on the edge $(x,xa)$ is $f(xa)$. The tree $f\rest x$ is the subtree rooted at $x\in\Ss$, where $f\rest x(y) = f(xy)$. For any coalgebra $(S,\dd)$, there is a unique coalgebra morphism $h:(S,\dd)\to((\Gs)^{\Sp},D)$ defined coinductively by
\begin{align*}
(h(s)\rest a,h(s)(a)) &= \letin{(t,z)}{\dd(s,a)}{(h(t),z)},
\end{align*}
where $s\in S$ and $a\in\Sigma$; equivalently,
\begin{align*}
h(s)(a) &= \snd (\delta(s,a)) & h(s)(ax) &= h(\fst (\delta(s,a)))(x),
\end{align*}
where $\fst$ and $\snd$ denote the projections onto the first and second components, respectively.

The coalgebraic view allows arbitrary protocols to inherit structure from the final coalgebra under $h^{-1}$, thereby providing a mechanism for transferring results on trees, such as entropy rate, to results on state transition systems.

There are other advantages as well. In this paper we have considered only \emph{homogeneous} measures on $\So$ and $\Go$, that is, those induced by \iid\ processes in which the probabilistic choices are independent and identically distributed, for finite $\Sigma$ and $\Gamma$. However, the coalgebraic definitions of protocol and reduction make sense even if $\Sigma$ and $\Gamma$ are countably infinite and even if the measures are non-homogeneous.

We have observed that a fixed measure $\mu$ on $\Sigma$ induces a unique homogeneous measure, also called $\mu$, on $\So$. But in the final coalgebra, we can go the other direction: For an arbitrary probability measure $\mu$ on $\So$ and state $f:\Sp\to\Gs$, there is a unique assignment of transition probabilities on $\Sp$ compatible with $\mu$, namely the conditional probability
\begin{align*}
f(xa) &= \frac{\mu(\set{\alpha}{xa\pref\alpha})}{\mu(\set{\alpha}{x\pref\alpha})}, 
\end{align*}
or 0 if the denominator is 0.
This determines the probabilistic behavior of the final coalgebra as a protocol starting in state $f$ when the input stream is distributed as $\mu$. This behavior would also be reflected in any protocol $(S,\dd)$ starting in any state $s\in h^{-1}(f)$ under the same measure on input streams, thus providing a semantics for $(S,\dd)$ even under non-homogeneous conditions.

In addition, as in Lemma \ref{lem:delta}\eqref{lem:deltaiii}, any measure $\mu$ on $\So$ induces a push-forward measure $\mu\circ(D^\omega)^{-1}$ on $\Go$. This gives a notion of reduction even in the non-homogeneous case. Thus we can lift the entire theory to Mealy automata that operate probabilistically relative to an arbitrary measure $\mu$ on $\So$. These are essentially discrete Markov transition systems with observations in $\Gs$.

Even more generally, one can envision a continuous-space setting in which the state set $S$ and alphabets $\Sigma$ and $\Gamma$ need not be discrete. The appropriate generalization would give reductions between discrete-time and continuous-space Markov transition systems as defined for example in \cite{Panangaden09,Doberkat07}.

As should be apparent, in this paper we have only scratched the surface of this theory, and there is much left to be done.

\section*{Acknowledgments}

Thanks to Swee Hong Chan, Bobby Kleinberg, Joel Ouaknine, Aaron Wagner for valuable discussions. Thanks to the anonymous referees for several suggestions for improving the presentation. Thanks to the Bellairs Research Institute of McGill University for providing a wonderful research environment. This research was supported by NSF grants CCF-1637532, IIS-1703846, IIS-1718108, and CCF-2008083, ARO grant W911NF-17-1-0592, and a grant from the Open Philanthropy project.


\end{document}